\PassOptionsToPackage{dvipsnames}{xcolor}
\PassOptionsToPackage{a4paper, bottom=0cm}{geometry}
\documentclass[prx,12pt,a4paper,onecolumn,margin=0.1cm,bottom=0.1cm,unpublished]{quantumarticle}
\pdfoutput=1
\usepackage[LGR,T1]{fontenc}
\usepackage[utf8]{inputenx}
\usepackage{comment}
\usepackage[english]{babel}
\usepackage[numbers,sort&compress]{natbib}

\usepackage{amsfonts, mathtools, amssymb, amsthm}
\usepackage{newtxtext}
\DeclareSymbolFont{symbolsC}{U}{npxsyc}{m}{n}
\SetSymbolFont{symbolsC}{bold}{U}{npxsyc}{b}{n}
\DeclareMathSymbol{\multimapdotboth}{\mathrel}{symbolsC}{22}

\newcommand{\dbot}{\mathbin{\text{$\bot\mkern-8mu\bot$}_{\text{CI}}}}
\let\oldperp\perp
\renewcommand{\perp}{\mathbin{\oldperp_{\text{d}}}}
\newcommand{\eperp}{\mathbin{\oldperp_{\text{e}}}}

\usepackage{makeidx}
\usepackage{xspace}

\usepackage{thm-restate}
\newcommand{\NLF}{{\textsc{NALF}}\xspace}
\newcommand{\interesting}{\textsc{Non-Algebraic}\xspace}
\newcommand{\interestingness}{\textsc{Non-Algebraicness}\xspace}
\newcommand{\noninteresting}{\textsc{Algebraic}\xspace}
\newcommand{\noninterestingness}{\textsc{Algebraicness}\xspace}
\usepackage{graphicx}
\usepackage{paralist}
\usepackage[normalem]{ulem}
\usepackage{tikz}
\usepackage{nicefrac}
\usepackage[dvipsnames]{xcolor}
\usepackage[page,title, titletoc]{appendix}
\usepackage[colorlinks=true, pdfborder={0 0 0}]{hyperref}
\definecolor{googleblue}{RGB}{34, 0, 204}
\definecolor{panblue}{RGB}{0,24,150}
\definecolor{carmine}{RGB}{150, 0, 24}
\hypersetup{
unicode=true,
bookmarksnumbered=false,
bookmarksopen=false,
breaklinks=false,
colorlinks=true,
linkcolor=carmine,
citecolor=googleblue,
urlcolor=panblue,
anchorcolor=OliveGreen}
\usepackage{array}
\usepackage{authblk}
\newcommand{\specialcell}[2][c]{%
\begin{tabular}[#1]{@{}c@{}}#2\end{tabular}}
\usepackage{makecell}

\usepackage{microtype}
\microtypecontext{spacing=nonfrench}
\microtypesetup{
expansion={true,nocompatibility},
protrusion={true,nocompatibility},
activate={true,nocompatibility},
tracking=false,
kerning=true,
spacing={true}
}

\usepackage[all=normal,floats=tight,mathspacing=tight,wordspacing=tight,paragraphs=normal,tracking=tight,charwidths=tight,mathdisplays=normal,sections=normal,margins=normal]{savetrees}

\makeatletter
\renewenvironment{proof}[1][\proofname]{\par
  \pushQED{\qed}
  \topsep6\p@\@plus6\p@\relax
  \trivlist
  \item[\hskip\labelsep
    \itshape #1\@addpunct{.}]\ignorespaces
}{\popQED\endtrivlist\@endpefalse}
\makeatother

\everypar=\expandafter{\the\everypar\loosness=-1 }
\newcommand{\undirectededge}{\multimapdotboth}

\usetikzlibrary{shapes,positioning,arrows.meta, calc}
\tikzset{
c1/.style={draw, regular polygon, regular polygon sides=3, minimum height=2em, inner sep=0, fill=GreenYellow}, 
q1/.style={draw, circle, minimum height=1.8em, inner sep=0, fill=Melon}, 
c/.style={draw, regular polygon, regular polygon sides=3, minimum height=2.8em, inner sep=0, fill=GreenYellow}, 
q/.style={draw, circle, minimum height=2.4em, inner sep=0, fill=Melon}, 
r/.style={draw, rectangle, minimum height=2em, fill= SpringGreen, inner sep=0.5em}, 
e/.style={{Stealth[scale=1.5]}-}, thick}

\tikzset{nv/.style={circle, color=red, fill=red, inner sep=0.5mm}}
\tikzset{rv/.style={circle, draw, thick, minimum size=7mm, inner sep=0.5mm}}
\tikzset{fv/.style={rectangle, draw, thick, minimum size=7mm, inner sep=0.5mm}}
\tikzset{lv/.style={circle, color=red, fill=gray!30, draw, thick, minimum size=7mm, inner sep=0.5mm}}
\tikzset{rve/.style={ellipse, draw, thick, minimum size=7mm, inner sep=0.5mm}}

\tikzset{rvs/.style={circle, draw, thick, minimum size=6mm, inner sep=0.5mm}}
\tikzset{fvs/.style={rectangle, draw, thick, minimum size=6mm, inner sep=0.5mm}}
\tikzset{lvs/.style={circle, color=red, fill=gray!30, draw, thick, minimum size=6mm, inner sep=0.5mm}}
\tikzset{rves/.style={ellipse, draw, thick, minimum size=6mm, inner sep=0.5mm}}

\tikzset{deg/.style={->, very thick, color=blue}}
\tikzset{degl/.style={->, very thick, color=red}}
\tikzset{beg/.style={<->, very thick, color=red}}
\tikzset{cdeg/.style={{Circle[length=+2pt 2.5,width=+2pt 2.5, fill=none]}->, very thick, color=blue}}
\tikzset{cceg/.style={{Circle[length=+2pt 2.5,width=+2pt 2.5, fill=none]}-{Circle[length=+2pt 2.5,width=+2pt 2.5, fill=none]}, very thick}}
\tikzset{uceg/.style={{Circle[length=+2pt 2.5,width=+2pt 2.5, fill=none]}-, very thick}}
\tikzset{ueg/.style={very thick}}

\usepackage{csquotes}

\newtheorem{theorem}{Theorem}
\newtheorem{conjecture}[theorem]{Conjecture}
\newtheorem{definition}[theorem]{Definition}
\newtheorem{corollary}[theorem]{Corollary}
\newtheorem{lemma}[theorem]{Lemma}
\newtheorem{remark}[theorem]{Remark}
\newtheorem{proposition}[theorem]{Proposition}

\makeatletter
\newcommand{\settheoremtag}[1]{%
  \let\oldthetheorem\thetheorem%
  \renewcommand{\thetheorem}{\oldthetheorem\hspace{0.6ex}[#1]}%
  \g@addto@macro\endtheorem{%
    \global\let\thetheorem\oldthetheorem}%
  }
\makeatother

\makeatletter
\newcommand*{\hyperlinkcite}[1]{\hyper@link{cite}{cite.#1\@extra@b@citeb}}%
\makeatother
\newcommand{\HLP}{\hyperlinkcite{HLP_2014}{\textsf{HLP}}\xspace}

\begin{document}
\title{Classifying Causal Structures: Ascertaining when Classical Correlations are Constrained by Inequalities}
\author{Shashaank Khanna}
\affiliation{Department of Mathematics,  University of York, Heslington, York, YO10 5DD, United Kingdom}
\author{Marina Maciel Ansanelli}
\affiliation{Perimeter Institute for Theoretical Physics, 31 Caroline Street North, Waterloo, Ontario Canada N2L 2Y5}
\affiliation{University of Waterloo, Waterloo, Ontario, Canada, N2L 3G1}
\author{Matthew F. Pusey}
\affiliation{Department of Mathematics,  University of York, Heslington, York, YO10 5DD, United Kingdom}
\author{Elie Wolfe}
\affiliation{Perimeter Institute for Theoretical Physics, 31 Caroline Street North, Waterloo, Ontario Canada N2L 2Y5}
\affiliation{University of Waterloo, Waterloo, Ontario, Canada, N2L 3G1}
\maketitle

\begin{abstract}
	The classical causal relations between a set of variables, some observed and some latent, can induce both equality constraints (typically conditional independencies) as well as inequality constraints (Instrumental and Bell inequalities being prototypical examples) on their compatible distribution over the observed variables. Enumerating a causal structure's implied inequality constraints is generally far more difficult than enumerating its equalities. Furthermore, only inequality constraints ever admit violation by quantum correlations. For both those reasons, it is important to classify causal scenarios into those which impose inequality constraints versus those which do not. Here we develop methods for detecting such scenarios by appealing to \mbox{\(d\)-separation}, \mbox{\(e\)-separation}, and incompatible supports. Many (perhaps all?) scenarios with exclusively equality constraints can be detected via a condition articulated by Henson, Lal and Pusey (HLP). Considering all scenarios with up to 4 observed variables, which number in the thousands, we are able to resolve all but three causal scenarios, providing evidence that the HLP condition is, in fact, exhaustive.
\end{abstract}

\newpage
\tableofcontents

\section{Introduction}

Many times in science we are not only interested in the observed correlations between events, but also in the underlying causal explanations that govern such observed phenomena. In the mathematical framework for causal inference~\cite{pearl, Verma1988, instrument, spirtes2001causation, lauritzen1996}, the candidates for those causal explanations are represented by directed acyclic graphs (DAGs), where each node is associated with a variable and each edge represents direct causal influence.

A DAG imposes causal compatibility constraints on the probability distributions that can be causally explained by it. For example, a probability distribution over variables $\{A, B, C\}$ where $A$ and $C$ remain correlated even after a value of $B$ is conditioned upon \emph{cannot} be causally explained by the DAG of Figure~\ref{fig_DAG_ABC}. All of the distributions over $\{A, B, C\}$ that can be explained by Figure~\ref{fig_DAG_ABC} need to satisfy $P(AC|B)=P(A|B)P(C|B)$.\footnote{We will denote probabilities by the concise notation $P(XY)\equiv P_{XY}(X=x,Y=y)$.}

\begin{figure}[h]
	\centering
	\begin{tikzpicture}
		\node[c](A) at (1,0){$A$};
		\node[c](B) at (3,0){$B$}
		edge[e] (A);
		\node[c](C) at (5,0){$C$}
		edge[e] (B);
	\end{tikzpicture}
	\caption{Example of DAG.}
	\label{fig_DAG_ABC}
\end{figure}
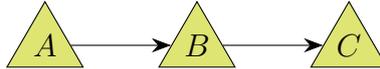

The causal compatibility constraint described above, denoted by $A\dbot C|B$, is a \emph{conditional independence relation}. That is, it says that one set of variables becomes independent of a second set of variables when conditioning on a third. In general, however, a DAG can impose more complicated types of constraints on the compatible probability distributions. This only happens for DAGs that have latent nodes, i.e. nodes associated with variables that do not appear in the probability distributions of interest.

From now on, the variables that do appear in the probability distributions of interest will be called \emph{observed variables}, while the ones that do not will be called \emph{latent variables}. Nodes associates with observed variables will be called \emph{observed nodes} and will be depicted by triangles, while nodes associated with latent variables will be called \emph{latent nodes} and will be depicted by circles.

An example of a DAG that imposes more complicated types of constraints on the compatible distributions over observed variables is the Bell DAG, presented in Figure~\ref{fig_Bell_DAG}. This DAG is of interest to physicists, as it encompasses the causal assumptions of Bell's theorem.

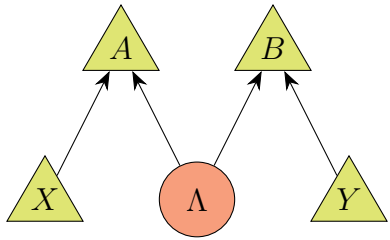
\begin{figure}[h]
	\centering
	\begin{tikzpicture}
		\node[q](A) at (2,0){$\Lambda$};
		\node[c](B) at (0,0){$X$};
		\node[c](C) at (4,0){$Y$};
		\node[c](D) at (3,2){$B$}
		edge[e] (A)
		edge[e] (C);
		\node[c](E) at (1,2){$A$}
		edge[e] (A)
		edge[e] (B);
	\end{tikzpicture}
	\caption{The Bell DAG. It encompasses the assumptions of Bell's theorem for a Bell Scenario where $X$ and $Y$ are the measurement settings of Alice and Bob, $A$ and $B$ are their outcomes and $\Lambda$ is a classical hidden variable. The probability distributions that are classically compatible with this DAG are those that decompose as in Equation~\eqref{eq_Bell_Markov}.}
	\label{fig_Bell_DAG}
\end{figure}

Bell's theorem~\cite{Bell, Bell1989} is central to the foundations of quantum mechanics~\cite{Norsen2017, Bricmont2016, Drr2020, PhysRev.47.777}, as it says that no locally causal hidden-variable theory in which the observers can choose their measurements independently of the source can ever be capable of reproducing all the operational predictions of quantum theory~\cite{Norsen2011, Norsen2009, Goldstein:2011, Norsen2007, Norsen2005}. These assumptions are encoded in the Bell DAG, as there the settings $X$ as $Y$ are not causally connected to the source $\Lambda$, as well as the setting in one wing does not causally influence the outcome in the other wing.

As it turns out, the Bell DAG imposes causal compatibility constraints on the compatible distributions over $\{A,B,X,Y\}$ that take the form of inequalities, which are precisely the Bell inequalities. This means that all the distributions which violate Bell's inequalities, including some of the quantum predictions for this scenario, \emph{cannot} be causally explained by the Bell DAG.

With the goal of causally explaining the violation of Bell's inequalities without changing the causal assumptions embedded in the structure of the Bell DAG, Henson, Lal and Pusey (\HLP)~\cite{HLP_2014} developed a generalization of Pearl's causal inference. In this generalized framework, the latent nodes can be associated with quantum or other generalized probabilistic theory (GPT) systems. \HLP also proved that the conditional independence constraints remain the same independently of the theory that describes the latent nodes; other types of constraints, like Bell's inequalities, can be violated.

The Bell DAG is just one of the many causal structures for which inequality constraints can be violated when the latent nodes are associated with nonclassical systems~\cite{HLP_2014, tobias, tobias2, chaves, PhysRevLett.123.140401}. In the Bell scenario, correlations that violate Bell inequalities have cryptographic applications precisely because of their non-classicality, so it seems reasonable to hope that other scenarios allowing non-classical correlations will have similar applications. Finding which causal structures imply inequalities which potentially admit quantum violation is a critical step towards such potential applications.

In \HLP, the concept of \interestingness of a causal structure was defined. There, this concept was called \emph{interestingness}.\footnote{
	Our use of \noninteresting and \interesting is equivalent to the distinction between \textsc{boring} and \textsc{interesting} causal structures in the terminology of Ref.~\cite{HLP_2014}. We recognize that, in general, renaming existing technical jargon is counterproductive. However, we feel that our terminology of \noninteresting and \interesting has semantic connotations so much closer to their underlying meanings as to warrant  the renaming. The authors acknowledge Robert W. Spekkens for suggesting these terms.} A causal structure is said to be \noninteresting when all of its causal compatibility constraints are of the form of conditional independence relations, even in the classical case. Conversely, if the causal structure imposes more complicated constraints on the compatible probability distributions, it is said to be \interesting. As proven in \HLP, only the \interesting scenarios are passive of witnessing a difference between the sets of classically and quantumly achievable probability distributions.

In \HLP, a sufficient condition for \noninterestingness of a causal structure was developed; it is referred to here as the \emph{HLP criterion}. A central motivation behind this work is that, at present, it is not known whether the \HLP criterion is also necessary for \noninterestingness. That is, if a DAG cannot be proven \noninteresting by virtue of the \HLP criterion, is the DAG necessarily \interesting? The conjecture that the \HLP criterion is indeed necessary for \noninterestingness will be referred to as the \emph{HLP conjecture}.

How might we evaluate the \HLP conjecture? To disprove it, we need to find only one DAG for which the \HLP criterion does not apply, but which can be proven \noninteresting by some other method. We did not pursue a search for such a counterexample, simply because we are unaware of any means to prove \noninterestingness when the \HLP criterion does not apply. We therefore concentrate on providing evidence in support of the conjecture being true. Namely, we show that as one considers ``larger and larger" DAGs, we can still certify the \interestingness of (almost) all DAGs for which the \HLP criterion does not apply.

To accomplish these goals we must clarify two preliminary questions. Firstly, how should we \emph{enumerate} over DAGs? One enumeration style --- the original enumeration choice employed by \HLP --- is to count DAGs by their total number of nodes. We can thus consider DAGs with 5 total nodes, with 6 total nodes, with 7 total nodes, etc. While this enumeration style has the advantage of simplicity, the arguably more natural enumeration which will be used here is to count by the total number of observed nodes. We can thus consider DAGs with 2 observed nodes, 3 observed nodes, 4 observed nodes, etc. From naive structural considerations alone, however, one might imagine that there are infinitely many DAGs with any fixed number of observed nodes. Motivated in-part by avoiding such infinities, when enumerating by the number of observed nodes we elect to work within Evans' framework of marginalized DAGs (mDAGs)~\cite{Evans2015}, which will be explained in Section~\ref{section_mDAG}.

The second preliminary question is how to prove that a DAG for which the \HLP criterion does not apply is indeed \interesting? Since we here seek to consider hundreds if not thousands of such DAGs, we are heavily invested in identifying \emph{algorithmic} techniques for proving \interestingness, within which we deprioritize computationally expensive methods. In stark contrast to the approach of \HLP, we extensively leverage a new result due to Evans~\cite{Evans_2022}, who has shown that a DAG $G$ is \noninteresting if and only if it is observationally equivalent to \emph{some} DAG that does not have latent variables --- where two DAGs are said to be (classically) observationally equivalent when their sets of (classically) compatible probability distributions are the same. As such, we herein primarily exploit necessary conditions for a graph to be observationally \emph{in}equivalent to every latent-free graph.

One such condition is that when the DAG is \emph{nonmaximal}, i.e. when it has a pair of nodes that are not adjacent (not connected by an arrow nor by a shared latent common cause) but are nevertheless not \(d\)-separated by any set of observed nodes, then the DAG is \emph{not} observationally equivalent to any latent-free DAG. Another condition says that for two DAGs to be observationally equivalent, they must admit the same \(e\)-separation relations over their observed nodes. A third condition says that for two DAGs to be observationally equivalent, they must admit the same set of compatible \emph{supports}. Although the latter condition subsumes the two former ones (as discussed in Section~\ref{subsec:supports}), the former conditions can be evaluated more efficiently. The latter condition involving supports (remarkably!) can be assessed using Fraser's algorithm~\cite{Fraser2020}, which generally requires higher computational overhead. All of these conditions can be utilized to prove the \interestingness of a given DAG, via proving the classical observational inequivalence of said DAG with \emph{every} latent-free graph which has the same number of observed nodes.

It is worth contrasting the tools we employ here to certify \interestingness with those employed in prior literature by \HLP and Pienaar~\cite{Pienaar2017}.\footnote{In doing so, we expose a critical error in one of the theorem's in Ref.~\cite{Pienaar2017} which we then rectify here (Happily, the handful of explicit DAGs which were declared as \interesting pursuant to the fallacious theorem in Ref.~\cite{Pienaar2017} are ultimately indeed \interesting nevertheless, and moreover their \interestingness follows from our replacement nonmaximality theorem here). This discussion is made in Appendix~\ref{app:pienaar}.} \HLP themselves attempted to explore all DAGs with 6 total nodes for which their sufficient condition for \noninterestingness did not apply. For all but five of these DAGs, \HLP proved \interestingness by means of deriving DAG-specific entropic inequalities and showing that those entropic inequalities could be violated by a DAG-specific distribution which nevertheless satisfied the conditional independence constraints of the DAG. One of the remaining five cases is the Bell DAG, which is long-since established as \interesting. Another one of the five is the so-called triangle scenario, whose \interestingness was proven using a one-off proof technique which \HLP did not generalize. The remaining three DAGs with 6 total nodes were eventually established as \interesting in a separate work by Pienaar~\cite{Pienaar2017}, who employed a proof technique using fine-grained entropic inequalities.

At first glance, the algorithmic proofs of \interestingness we employ here may seem unrelated to those utilized by \HLP or Pienaar. However, our theorem relating \(e\)-separation to \interestingness turns out to recover all but one of the \interestingness results that \HLP achieved by appealing to entropic inequalities. Additionally, our application of Fraser's algorithm further witnesses the \interestingness of every other DAG conjectured by \HLP to be \interesting, including the three of which were only \emph{proven} \interesting in the later work of Pienaar~\cite{Pienaar2017}.
We are therefore confident that our plethora of techniques likely supersede those earlier employed by \HLP and Pienaar~\cite{Pienaar2017}, despite not having a formal proof yet.

We summarize our ultimate findings in Table~\ref{tab:introsummary}. We interpret these results as evidence in favour of the \HLP conjecture: among the thousands of analyzed mDAGs, there are only 3 potential counterexamples.

\begin{table}[h!]
	\centering
	\renewcommand{\arraystretch}{1.5}
	\begin{tabular}{|p{8cm}|c|c|}
		\hline
		\small Category & \specialcell{\small mDAGs with \\\textbf{3} observed nodes} & \specialcell{\small mDAGs with \\\small \textbf{4} observed nodes}\\
		\hline
		{\small Total Count}                                                                & 46 & 2809 \\
		{\small  remaining \# for which the \HLP criterion does not apply}                  & 5  & 996  \\
		{\small  remaining \# for which our nonmaximality condition does not apply}         & 1  & 186  \\
		{\small  remaining \# for which our setwise nonmaximality condition does not apply} & 0  & 78   \\
		{\small  remaining \# for which our \(d\)-separation condition does not apply}      & 0  & 60   \\
		{\small  remaining \# not resolved by our use of Fraser's algorithm}                & 0  & 3    \\
		\hline
	\end{tabular}
	\caption{A summary of our findings: apart from the HLP criterion, which shows \noninterestingness, all of the other conditions listed show \interestingness. Note that here we are counting by \emph{unlabelled} DAGs. That is, two labelled DAGs which are equivalent under a relabelling of the observed nodes and/or a relabelling of the hidden nodes are represented by a single unlabelled DAG in these enumerations. }\label{tab:introsummary}
\end{table}

\begin{sloppypar}
	As discussed, this work is of interest to quantum physicists because the \interesting DAGs are the possible candidates to explore quantum advantages in device independent information processing protocols~\cite{acin2007device, vazirani2019fully}. These DAGs are also the ones that should be looked into to compare quantum theory to more general probabilistic theories (GPTs)~\cite{PhysRevA.75.032304, plavala2021general}.
\end{sloppypar}

On the other hand, our problem is also of central interest for purely classical causal inference. The set of probability distributions which are classically compatible with an \noninteresting DAG is constrained only by conditional independence relations, which can be obtained from the \(d\)-separation relations of the DAG. Isolating a sufficient set of \(d\)-separation relations in a graph\footnote{A set of \(d\)-separation relations is said to be sufficient for a DAG if all other \(d\)-separation relation in the DAG can be inferred from the sufficient subset by application of semigraphoid axioms.} is a well-studied problem. It is of paramount value to a classical data scientist, therefore, to know if the causal hypothesis encoded in a DAG may or may not be falsified by accounting for nontrivial inequality constraints. Such inequality constraints, when present, are often difficult to explicitly characterize.

\subsection*{Structure of the paper}

In Section~\ref{sec_preliminary_causal_structures} we present an introduction to the formalism used in Causal Inference, and proceed to explain what \interesting and \noninteresting precisely mean in this context in Section~\ref{sec:interestingness}. We state and prove our \(e\)-separation condition in Section~\ref{methods_for_interestingness}. In Section~\ref{sec:results} we present our computational results and also discuss the methods we tried to confirm the \interestingness of the three left mDAGs of 4 observed nodes. Our conclusions can be found in Section~\ref{con}.

\section{Preliminaries}

\subsection{Causal Explanations of Observational Data}
\label{sec_preliminary_causal_structures}

The area of causal inference is concerned with finding potential causal explanations for observed events. For example, imagine that we want to find out what is the causal relationship between three events: a cloudy sky, rain and the floor being wet. Figure~\ref{fig_clouds} depicts two possible causal structures between these three events; in~\ref{fig_clouds}(a) we hypothesize that the clouds cause the rain and the rain makes the floor wet. In~\ref{fig_clouds}(b), on the other hand, we hypothesize that the wet floor causes the rain, and the rain makes the sky cloudy.

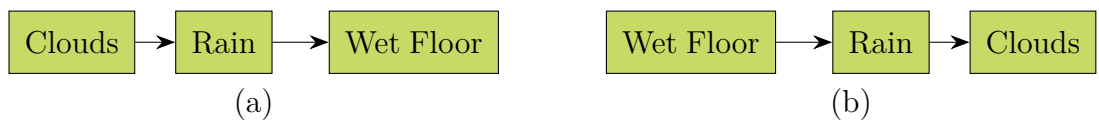
\begin{figure}[h]
	\hspace{-10pt}
	\setlength{\tabcolsep}{20pt}
	\begin{tabular}{c  c}
		\begin{tikzpicture}
			\node[r](A) at (0,0){Clouds};
			\node[r](B) at (2,0){Rain}
			edge[e] (A);
			\node[r](C) at (4.5,0){Wet Floor}
			edge[e] (B);
		\end{tikzpicture}

		    &
		\begin{tikzpicture}
			\node[r](A) at (0,0){Wet Floor};
			\node[r](B) at (2.5,0){Rain}
			edge[e] (A);
			\node[r](C) at (4.5,0){Clouds}
			edge[e] (B);
		\end{tikzpicture} \\
		(a) & (b)
	\end{tabular}
	\caption{Two hypothesis for the causal relationships between observing clouds, rain and wet floor. By intervening on the experiment we can attest that (a) is the correct explanation, but we cannot attest that if we only passively look at the correlations between those events.}
	\label{fig_clouds}
\end{figure}

There is an easy way we can check that Figure~\ref{fig_clouds}(b) is not the correct causal hypotheses: if we pour water on the floor in a sunny day, it will not start raining.

Note that this method presupposes that we can \emph{intervene} on our experiment, meaning that we can force one of the variables of the experiment (wet floor) to assume the value we want. However, it is not always possible to do that; sometimes it is unethical or there are technical or fundamental limitations to do so.

If the experimenter cannot make interventions on the variables of interest, it is still possible for them to draw some conclusions from a passive observation of the correlations between the events of interest. As we will see, each causal structure imposes constraints on which probability distributions obtained from passive observations can be classically explained by it. These constraints can be tested against the observed probability distribution to see if the given causal structure is a valid causal hypothesis for the observed phenomena. In Figure~\ref{fig_clouds} it happens that both causal structures impose the same constraints on probability distributions, so they are not distinguishable by passive observations.

In this work we will be mainly concerned with passive observations. The sets defined in Section~\ref{sec:interestingness} that are of relevance to us make reference to distributions obtained from passive observations only.

\subsection{Directed Acyclic Graphs}

In the framework we use here, the mathematical object that describes a causal structure is a directed acyclic graph (DAG). A directed graph $G$ is a pair $(\boldsymbol{A}, \boldsymbol{E})$, where $\boldsymbol{A}$ is a set of nodes and $\boldsymbol{E} \subseteq \boldsymbol{A}\times \boldsymbol{A}$ is a set of directed edges. A directed acyclic graph (DAG) is a directed graph that has no directed cycles. Below, we introduce some definitions regarding DAGs that will be useful later.

\begin{definition}[\textbf{Children, Parents, Descendants, Ancestors}]

	Let $X$ be a node of a DAG $G$. If $Y$ is another node of $G$ such that there is a directed edge $X \rightarrow Y$, then $Y$ is called a \underline{child} of $X$, and $X$ is called a \underline{parent} of $Y$. The set of all children of $X$ is denoted as $\textrm{CH}_G(X)$, and the set of all parents of $X$ is denoted as $\textrm{PA}_G(X)$.

	A directed path is a sequence of nodes $X^1, X^2, X^3...... X^n$ such that $X^i \rightarrow X^{i+1}$ for $i=1,...,n$. The \underline{descendants} of $X$ in $G$ are all the nodes in $G$ that can be reached from $X$ by a directed path. Conversely, all the nodes that have $X$ as a descendant are called \underline{ancestors} of $X$. The set of all ancestors of $X$ is denoted as $\textrm{ANC}_G(X)$, meanwhile the set of all descendants of $X$ can be denoted as $\textrm{DES}_G(X)$.\footnote{In this article we follow the convention of Refs.~\cite{Richardson1997, Steudel2015, vanderZander2019} and others, namely, the convention in which $X\in \textrm{ANC}_G(X)$ and $X\in \textrm{DEC}_G(X)$ but $X\not\in \textrm{PA}_G(X)$ and $X\not\in \textrm{CH}_G(X)$. That is, a node is considered it own ancestor and its own descendant, but not its own parent or its own child.}
\end{definition}

For example, in the DAG of Figure~\ref{fig_example_DAG} we have that $\text{PA}_G(B)=\{A,D\}$, and that $E$ is a descendant of $D$, even if it is not its child.

\begin{figure}[h]
	\centering
	\begin{tikzpicture}
		\node[c](A) at (1,0){$D$};
		\node[c](B) at (-1,2){$A$}
		edge[e] (A);
		\node[c](D) at (3,2){$C$}
		edge[e] (A);
		\node[c](C) at (5,2){$E$}
		edge[e] (D);
		\node[c](E) at (1,2){$B$}
		edge[e] (A)
		edge[e] (B);
	\end{tikzpicture}
	\caption{Example of a directed acyclic graph (DAG). The probability distributions that are classically compatible with this DAG are those that can be decomposed as in Equation~\eqref{eq_markov_example}.}
	\label{fig_example_DAG}
\end{figure}
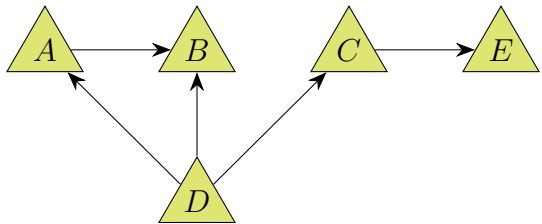

\subsection{DAGs as Causal Structures}
\label{cbn}
When we associate each node of a DAG with an event of interest, the DAG is a representation of a causal structure: an edge $X \rightarrow Y$ shows a possibility of a direct causal influence of $X$ on $Y$. Here, we will indicate the variable associated with a node by the same capital letter as the node itself. If all the events of interest are described by classical random variables, the idea that a probability distribution ``can be causally explained" by a certain DAG is formalized through the \emph{Markov condition}:

\begin{definition}[\textbf{Markov Condition}]
	\label{markov}
	Let $G$ be a DAG with nodes $\boldsymbol{A}$.
	A probability distribution $P$ over the variables $\boldsymbol{A}$ is said to be \emph{Markov with respect to $G$} if  it can be factorized as:
	\begin{equation}
		P(\boldsymbol{A})=\prod_i P\left(X^i|\text{PA}_G\left(X^i\right)\right)
	\end{equation}

	Where $\boldsymbol{A}=\{X^1,...,X^n\}$ and $\text{PA}_G\left(X^i\right)$ is the set of parents of the node $X^i$ in $G$.

	If $P$ is Markov with respect to $G$, we also say that it is `` classically compatible" with $G$.
\end{definition}

As an example of this definition, a joint probability distribution $P_{ABCDE}$ over the random variables $A$, $B$, $C$, $D$ and $E$ is Markov with respect to the DAG of Figure~\ref{fig_example_DAG} if it can be decomposed as:

\begin{equation}
	\label{eq_markov_example}
	P(ABCDE) = P(A|D)P(B|A,D)P(C|D)P(D)P(E|C)
\end{equation}

Therefore, through the Markov condition, each DAG imposes constraints on the probability distributions that are classically compatible with it. A DAG which is classically compatible with \emph{every} probability distribution, that is, a DAG that does not impose any constraint on the classically compatible distributions, is said to be a \emph{saturated} DAG.

\subsection{\(d\)-separation: A graphical criterion for conditional independence}

If $P$ is a probability distribution over a certain set of random variables $\boldsymbol{A}$, we say that the variables of the subset $\boldsymbol{X}\subseteq \boldsymbol{A}$ are \emph{conditionally independent} of the variables of the subset $\boldsymbol{Y}\subseteq \boldsymbol{A}$ given the subset $\boldsymbol{Z}\subseteq \boldsymbol{A}$ if $P$ can be factorized as $P(\boldsymbol{X},\boldsymbol{Y}|\boldsymbol{Z})=P(\boldsymbol{X}|\boldsymbol{Z})P(\boldsymbol{Y}|\boldsymbol{Z})$. This is denoted by $\boldsymbol{X}\dbot\boldsymbol{Y}|\boldsymbol{Z}$.

Some of the constraints that a DAG imposes on the probability distributions that are classically compatible with it are of the form of conditional independence: if a probability distribution $P$ \emph{cannot} be factorized according to a certain conditional independence that is imposed by the DAG, then it is not classically compatible with the DAG. As it turns out, there exists a graphical algorithm to obtain all the conditional independence constraints that are imposed by a DAG. This algorithm is called ``\(d\)-separation", and we will now describe it.

If $G$ is a DAG with nodes $\boldsymbol{A}$ and $\boldsymbol{X},\boldsymbol{Y},\boldsymbol{Z}\subseteq \boldsymbol{A}$ are sets of nodes in $G$, the \(d\)-separation algorithm says whether $\boldsymbol{X}$ and $\boldsymbol{Y}$ are ``\(d\)-separated" by $\boldsymbol{Z}$. This happens if all the undirected paths (paths that ignore the direction of the arrows) from $\boldsymbol{X}$ to $\boldsymbol{Y}$ are \emph{blocked} by $\boldsymbol{Z}$. A path is blocked if one or more of the following is true:
\begin{enumerate}
	\item There is a chain of nodes along the path:  $i\rightarrow m\rightarrow j$ such that $m\in Z$.
	\item There is a fork along the path:   $i\leftarrow m\rightarrow j$ such that $m\in Z$.
	\item There is a collider along the path:  $i\rightarrow m\leftarrow j$ such that $m\notin Z$ and $d\notin Z$ for all the descendants $d$ of $m$.
\end{enumerate}
If $X$ is \(d\)-separated of $Y$ given $Z$ in the DAG under consideration, we denote it as $X \perp Y |Z$.

For example, for the DAG of Figure~\ref{fig_DAG_ABC} the \(d\)-separation criterion says that $A\perp C|B$. This means that the event $A$ should become independent of the event $C$ upon knowledge of $B$; if your distribution $P$ \emph{does not} satisfy the constraint $P(AC|B)=P(A|B)P(C|B)$, Figure~\ref{fig_DAG_ABC} is \emph{not} a valid causal explanation for it. Interpreting the variables $\{A,B,C\}$ as respectively the clouds, rain and wet floor (such as in Figure~\ref{fig_clouds}(a)), this \(d\)-separation relation says that the occurrence of clouds becomes independent of the floor being wet if we already know whether it is raining.

The following theorem, proven in Ref.~\cite{Verma1988}, makes explicit the connection between \(d\)- separation relations and conditional independence relations:

\begin{theorem}[\(d\)-separation and conditional independence]
	\label{th1}
	Let $G$ be a DAG with nodes $\boldsymbol{A}$, and let $\boldsymbol{X}\subseteq \boldsymbol{A}$, $\boldsymbol{Y}\subseteq \boldsymbol{A}$ and $\boldsymbol{Z}\subseteq \boldsymbol{A}$ be three disjoint subsets of $\boldsymbol{A}$. Then:

	\begin{enumerate}
		\item If $G$ has the \(d\)-separation relation $\boldsymbol{X}\perp \boldsymbol{Y}|\boldsymbol{Z}$, then all of the probability distributions over the variables $\boldsymbol{A}$ which are Markov with respect to $G$ need to satisfy $\boldsymbol{X}\dbot \boldsymbol{Y}|\boldsymbol{Z}$.
		\item If $G$ \emph{does not have} the \(d\)-separation relation $\boldsymbol{X}\perp \boldsymbol{Y}|\boldsymbol{Z}$, then there exists some probability distribution over the variables $\boldsymbol{A}$ which is Markov with respect to $G$ and \emph{does not} satisfy $\boldsymbol{X}\dbot \boldsymbol{Y}|\boldsymbol{Z}$.
	\end{enumerate}
\end{theorem}

When a DAG does not have latent nodes, the \emph{only} constraints that it imposes on the compatible distributions are the conditional independence relations associated with the \(d\)-separation relations of the DAG~\cite[Theorem 3]{dsep_complete}:
\begin{theorem}[\(d\)-separation is complete for DAGs without latent nodes]
	\label{th_latent_free}
	Let $G$ be a DAG with nodes $\boldsymbol{A}$. A probability distribution over the variables $\boldsymbol{A}$ is Markov with respect to $G$ if and only if it satisfies all of the conditional independence relations associated with the \(d\)-separation relations of $G$.
\end{theorem}

Importantly, in general this is \emph{not} true for DAGs that include latent nodes, as we will see now.

\subsection{Latent-Permitting DAGs}
\label{section_latents}

As discussed in the introduction, sometimes we want to allow for causal explanations that include \emph{latent variables}, i.e. variables that do not appear in the final probability distribution we are trying to explain. When our DAG of interest can have latent nodes, we call it a \emph{latent-permitting DAG}, as opposed to the \emph{latent-free DAGs} that only have observed nodes.

Let $G$ be a DAG that has the set of nodes $\boldsymbol{A}=\boldsymbol{V}\cup \boldsymbol{L}$, $\boldsymbol{V}$ and $\boldsymbol{L}$ disjoint, where $\boldsymbol{V}$ are the observed nodes and $\boldsymbol{L}$ are the latent nodes. Mimicking the terminology used for the latent-free case, we will say that a probability distribution $P(\boldsymbol{V})$ over the observed variables is \emph{classically compatible with $G$} if there exists some probability distribution $P(\boldsymbol{V}, \boldsymbol{L})$ over $\boldsymbol{V}\cup \boldsymbol{L}$ such that:
\begin{itemize}
	\item $P(\boldsymbol{V}, \boldsymbol{L})$ is Markov with respect to $G$.
	\item The marginal of $P(\boldsymbol{V}, \boldsymbol{L})$ over $\boldsymbol{V}$ is the original probability distribution that we are interested in, i.e. $\sum_{\boldsymbol{L}}P(\boldsymbol{V}, \boldsymbol{L})=\sum_{\boldsymbol{L}}P(\boldsymbol{V}|\boldsymbol{L})P(\boldsymbol{L})=P(\boldsymbol{V})$.
\end{itemize}

As an example, a probability distribution $P_{ABXY}$ over the random variables $A$, $B$, $X$ and $Y$ is classically compatible with the Bell DAG (Figure~\ref{fig_Bell_DAG}) if it can be decomposed as:

\begin{equation}
	\label{eq_Bell_Markov}
	P(ABXY) = \sum_\Lambda P(A|X,\Lambda)P(B|Y,\Lambda)P(X)P(Y)P(\Lambda)
\end{equation}

When dealing with latent-permitting DAGs, we will call the conditional independence relations that only involve observed variables the \emph{observed conditional independence relations}, and similarly the \(d\)-separation relations that only involve observed nodes will be called \emph{observed \(d\)-separation relations}. From the definition of classical compatibility with a latent-permitting DAG, we can see that the conclusions of Theorem~\ref{th1} are also valid for latent-permitting DAGs:

\begin{theorem}[Observed \(d\)-separation in latent-permitting DAGs]
	\label{th_dsep_latent_permitting}
	Let $G$ be a DAG with nodes $\boldsymbol{A}=\boldsymbol{V}\cup\boldsymbol{L}$, where $\boldsymbol{V}$ are observed nodes and $\boldsymbol{L}$ are latent nodes. Let $\boldsymbol{X}\subseteq \boldsymbol{V}$, $\boldsymbol{Y}\subseteq \boldsymbol{V}$ and $\boldsymbol{Z}\subseteq \boldsymbol{V}$ be three disjoint sets of observed nodes of $G$. Then:

	\begin{enumerate}
		\item If $G$ has the \(d\)-separation relation $\boldsymbol{X}\perp \boldsymbol{Y}|\boldsymbol{Z}$, then all of the probability distributions over the variables $\boldsymbol{V}$ which are classically compatible with $G$ need to satisfy $\boldsymbol{X}\dbot \boldsymbol{Y}|\boldsymbol{Z}$.
		\item If $G$ \emph{does not have} the \(d\)-separation relation $\boldsymbol{X}\perp \boldsymbol{Y}|\boldsymbol{Z}$, then there exists some probability distribution over the variables $\boldsymbol{V}$ which is classically compatible with $G$ and \emph{does not} satisfy $\boldsymbol{X}\dbot \boldsymbol{Y}|\boldsymbol{Z}$.
	\end{enumerate}
\end{theorem}

As seen in Theorem~\ref{th_latent_free}, the only constraints that latent-free DAGs impose on the compatible probability distributions are the conditional independence relations, that can be obtained from \(d\)-separation. Consequently, if a DAG $G$ has nodes $\boldsymbol{A}=\boldsymbol{V}\cup \boldsymbol{L}$, all the constraints that it imposes on the compatible joint probability distributions $P(\boldsymbol{A})=P(\boldsymbol{V}, \boldsymbol{L})$ are the conditional independence relations. However, if we are interested only on distributions over the observed variables $\boldsymbol{V}$, sometimes the conditional independence constraints that involve the latent variables $\boldsymbol{L}$ might induce more complicated extra constraints on the distributions over the observed variables $\boldsymbol{V}$.  If a probability distribution $P(\boldsymbol{V})$ over the observed variables satisfies both the observed conditional independence relations (obtained from \(d\)-separation) and the extra constraints derived from the conditional independence relations that involve the latent variables $\boldsymbol{L}$, then it is classically compatible with $G$.

Therefore, in principle one could find all the \(d\)-separation relations of a DAG (involving both observed and latent nodes), thus getting conditional independence constraints on $P(\boldsymbol{V},\boldsymbol{L})$, and from there infer the constraints on the probability distribution over the observed variables, $P(\boldsymbol{V})$. This process will be exemplified with the Bell DAG in the beginning of the next section. Inferring constraints on $P(\boldsymbol{V})$ from the conditional independence relations of $P(\boldsymbol{V},\boldsymbol{L})$, however, is in general very complicated.

\section{The problem: Looking for \interesting Scenarios}
\label{sec:interestingness}

The goal of this work is to classify DAGs as \noninteresting or not, following the concept introduced in \HLP. Before establishing this concept in full generality, we will explore the example of the Bell DAG (Figure~\ref{fig_Bell_DAG}).

As discussed in Section~\ref{section_latents}, the \(d\)-separation criterion gives us all the classical compatibility constraints imposed on the probability distribution over \emph{all the nodes}, $P(\boldsymbol{V},\boldsymbol{L})$. In the case of the Bell DAG, we have $P(\boldsymbol{V},\boldsymbol{L})=P(ABXY\Lambda)$, while $P(\boldsymbol{V})=P(ABXY)$.

The conditional independence relations of $P(ABXY\Lambda)$, that come from the \(d\)-separation relations of the Bell DAG, are:
\begin{align}
	P(\Lambda|XY)   & =P(\Lambda) \label{eq_desep_bell_1}                              \\
	P(AB|XY\Lambda) & =P(A|X\Lambda)P(B|Y\Lambda) \label{eq_desep_bell_2}              \\
	                & \hspace{-13ex}\begin{cases} P(A|XY) & =P(A|X)   \\
		P(B|XY) & =P(B|Y)   \\
		P(XY)   & =P(X)P(Y)\end{cases} \label{eq_desep_bell_3}
\end{align}

These equations will give rise to the constraints imposed by the Bell DAG on $P(ABXY)$ through the marginalization of $\Lambda$. Equations~\eqref{eq_desep_bell_3}, that do not involve $\Lambda$, are automatically transported as observed conditional independence constraints imposed by the Bell DAG on $P(ABXY)$. Equations~\eqref{eq_desep_bell_1} and~\eqref{eq_desep_bell_2}, that do involve the latent variable $\Lambda$, will give rise to another type of constraint on $P(ABXY)$: the Bell's inequality. In fact, Equation~\eqref{eq_desep_bell_1} encodes the no-superdeterminism assumption and Equation~\eqref{eq_desep_bell_2} encodes the local causality assumption that are used to derive Bell's inequality.

Therefore, to study the Bell DAG, it is \emph{not enough} to just look at the conditional independence constraints on the compatible distributions $P(ABXY)$ (Equations~\eqref{eq_desep_bell_3}). If one does that, they would miss important information that is encoded in Bell's inequality. This is so because the set of probability distributions $P(ABXY)$ that satisfy only the conditional independence relations of Equations~\eqref{eq_desep_bell_3} is \emph{strictly larger} than the set of probability distributions that satisfies these conditional independence relations \emph{and} Bell's inequality. In other words, Bell's inequality is not implied by Equations~\eqref{eq_desep_bell_3}.

This is the core of the concept of \interestingness: a \interesting DAG imposes nontrivial inequality constraints on the classically compatible distributions.  A ``nontrivial" inequality constraint is an inequality that is not implied by the observed conditional independence relations of the DAG, along with nonnegativity of all probabilities and normalization.\footnote{There is another type of equality constraint that a DAG can impose on the compatible distributions, apart from conditional independence relations: the ``nested Markov constraints". In Ref.~\cite{Evans_2022} it was proven that every DAG that presents nested Markov equality constraints also present nontrivial inequalities. Consequently, \interestingness may be equivalently defined relative to \emph{all implied equality constraints} or relative to \emph{all implied conditional independence relations}.} Note that an \noninteresting DAG can impose \emph{trivial} inequality constraints on the compatible distributions: for example, if the node $\Lambda$ in the Bell DAG was treated as an observed node, then the Bell inequalities would still be satisfied by the compatible distributions $P(ABXY\Lambda)$. However, in this case the Bell inequalities would be trivial, because they would be implied by \emph{observed} conditional independence relations (Eqs.~\eqref{eq_desep_bell_1} and~\eqref{eq_desep_bell_2}).

To formalize the idea of \interestingness, we will introduce a few definitions:

\begin{definition}
	\label{def2}
	Let $G$ be a DAG. The sets \textit{$\mathcal{C}_G$} and and \textit{$\mathcal{I}_G$}  of probability distributions over observed variables are defined as follows:

	\begin{enumerate}
		\item \textit{$\mathcal{C}_G$}: Set of probability distributions that are classically compatible with $G$.
		\item \textit{$\mathcal{I}_G$}: Set of probability distributions that satisfy all the conditional independence constraints that follow from the observed \(d\)-separation relations of $G$.
	\end{enumerate}
\end{definition}

For the case of the Bell DAG, $\mathcal{I}_{\text{Bell}}$ represents the set of distributions that obey the Equations~\eqref{eq_desep_bell_3}. By contrast, $\mathcal{C}_{\text{Bell}}$ consists of a strict subset of $\mathcal{I}_{\text{Bell}}$, where we additionally restrict the conditional probabilities $P(AB|XY)$ to satisfy Bell's inequalities and thus lie in the local polytope. %

Theorem~\ref{th_dsep_latent_permitting} shows that $\mathcal{C}_G \subseteq \mathcal{I}_G$ for every DAG $G$. This is so because all the probability distributions that are classically compatible with $G$ have to satisfy the conditional independence constraints that come from its observed \(d\)-separation relations. When the observed conditional independence relations are the only constraints imposed by a DAG on the compatible probability distributions over observed variables, the DAG is said to be \noninteresting:

\begin{definition}[\textbf{\noninterestingness}]
	Let $G$ be a DAG. If $\mathcal{C}_G = \mathcal{I}_G$, then $G$ is said to be \noninteresting. Conversely, if $\mathcal{C}_G \subsetneq \mathcal{I}_G$, then $G$ is said to be \interesting.
\end{definition}
We borrow the terminology of \interesting and \noninteresting from algebraic geometry: An algebraic set is defined by polynomial equalities (or more generally, by some finite union of sets each of which is defined by polynomial equalities). Semialgebraic sets, by contrast, are characterised by both polynomial equalities and polynomial inequalities. To emphasise that a DAG's set of (classically) compatible distributions is defined by \emph{more} that just the conditional independence (notably, equality) constraints, we therefore elect to speak of such a DAG as \interesting.

As proven by \HLP, the observed conditional independence constraints imposed by a DAG on the compatible probability distributions do not change when the latent variables of the DAG are associated with quantum systems or other GPT systems. Therefore, if one is interested in studying causal structures that provide any quantum or GPT observational advantage, then there is only hope among the \interesting scenarios. If a DAG is \noninteresting, then \emph{all} the probability distributions that exhibit the conditional independence relations associated with its observed \(d\)-separation relations can be explained by this DAG \emph{classically}.

Theorem~\ref{th_latent_free} implies that every latent-free DAG is \noninteresting. In \HLP, a stronger sufficient condition for \noninterestingness is provided, together with an algorithmic strategy to check it. This condition, called the \emph{HLP criterion}, will be discussed in Section \ref{sec_HLP_criterion}. It is still not known whether the \HLP criterion is also necessary for \noninterestingness; its possible outcomes for a given DAG are either that it is \noninteresting or that it is ``unresolved". The unresolved DAGs thus need to be assessed by some other method.

Based on the \HLP criterion and certain types of entropic inequalities, \HLP and Pienaar~\cite{Pienaar2017} classified the \noninterestingness or \interestingness of all DAGs of up to 6 \emph{total} nodes (observed and latent), thus leaving no DAGs of 6 total nodes with inconclusive status. In this work, however, we elect to count DAGs not by their \emph{total} node count but rather by their number of \emph{observed} nodes. It turns out that \HLP's complete classification of DAGs with \emph{total} node count up to 6 meant they resolved \emph{all} DAGs with 3 observed nodes, a \emph{few} DAGs with 4 observed nodes, and \emph{no} DAGs with 5 or more observed nodes. Here we attempt to tackle the \noninterestingness classification of all causal structures with up to 4 observed nodes.\footnote{Note that the set of all DAGs with 4 observed nodes \emph{includes} the set of all DAGs with 7 total nodes which persist under \HLP's reduction techniques. Thus, this work can \emph{also} be considered an extension of \HLP's classification from up-to-6 to up-to-7 total nodes. As noted in Section~\ref{sec:results}, we ulitmately resolve the \noninterestingness or \interestingness of \emph{all but three} DAGs of 4 observed nodes. Only \emph{one} of those 3 has 7 total nodes. Thus, we ultimately resolve \emph{all but one} DAG with 7 total nodes.} To do so, we will utilize the \emph{mDAG formalism} introduced by Evans in Ref.~\cite{Evans2015}.

\subsection{Simplifying the problem by using mDAG formalism}
\label{section_mDAG}

Two DAGs $G$ and $H$ are said to be \emph{classically observationally equivalent} when $\mathcal{C}_G=\mathcal{C}_H$. Note that $ \mathcal{C}_G= \mathcal{C}_H$ implies that $\mathcal{I}_G=\mathcal{I}_{H}$: by Theorem~\ref{th_dsep_latent_permitting}, if a certain \(d\)-separation relation is \emph{not} present in a DAG, it is always possible to find a probability distribution that violates the conditional independence corresponding to this \(d\)-separation relation  and is classically compatible with the DAG. In other words, if two DAGs are classically compatible with the same sets of distributions, they have to present the same \(d\)-separation relations.\footnote{On the other hand, if DAGs $G$ and $H$ are such that $\mathcal{I}_G=\mathcal{I}_H$, this does \emph{not} imply that  $\mathcal{C}_G=\mathcal{C}_H$.}

In short,\begin{equation}
	\mathcal{C}_G = \mathcal{C}_H \implies \mathcal{I}_G = \mathcal{I}_H,\end{equation}
with contrapositive
\begin{equation}
	\mathcal{I}_G \neq \mathcal{I}_H \implies \mathcal{C}_G \neq \mathcal{C}_H.
	\label{eq_different_I}
\end{equation}

In particular, this means that if a DAG $G$ is \noninteresting (\interesting), then all of the DAGs $H$ that are compatible with it are also \noninteresting (\interesting).

In this section, we will present two results of~\cite{Evans2015} that prove classical observational equivalence, thus reducing the number of DAGs that have to be examined for \noninterestingness. After presenting the two results, we will show that they allow for a definition of a new object, called mDAG, that encompasses this simplification.

To do so, we start with the definition of \emph{exogenization}. It might be easier to understand this definition by following Figure~\ref{fig3}, where DAG~\ref{fig3}(b) is obtained from DAG~\ref{fig3}(a) by exogenizing node $B$.

\begin{definition}[\textbf{Exogenized DAG}]
	\label{exg}
	Let $G$ be a DAG and let $\lambda$ be a latent variable of $G$. We define the \textit{exogenized DAG} $\mathcal{E}(G,\lambda)$ as follows: take the vertices and edges of $G$ and (1) add an edge $m\rightarrow n$ from every $m\in \text{PA}_G(\lambda)$ to every $n \in CH_G(\lambda)$ and (2) delete edges $m\rightarrow \lambda$ for every $l \in \text{PA}_{G}(\lambda)$. All other edges remain the same. \end{definition}

With this definition at hand, we state the Lemma 3.7 of~\cite{Evans2015}:

\begin{lemma}[Exogenization]
	\label{lemma_exogenization}
	Let $G$ be DAG with observed nodes $V$ and latent nodes $L$, and let $\lambda\in L$ be a latent node of $G$. Furthermore, let $\Tilde{G}=\mathcal{E}(G,\lambda)$. Then,  $\mathcal{C}_G=\mathcal{C}_{\Tilde{G}}$.
\end{lemma}

Now, we state the Lemma 3.8 of~\cite{Evans2015}. This Lemma is also illustrated in Figure~\ref{fig3}, where it is applied to go from~\ref{fig3}(b) to~\ref{fig3}(c).

\begin{lemma}[No redundant latents]
	\label{lemma_no_redundancy}
	Let $G$ be a DAG with observed nodes $V$ and latent nodes $L$. Let $\lambda\in L$ and $\mu\in L$ be latent nodes of $G$ such that $\lambda \neq \mu$, $\text{PA}_G(\lambda)=\text{PA}_G(\mu)=\emptyset$ and $CH_G(\lambda) \subseteq CH_G(\mu)$. In this case, we say that the node $\lambda$ is ``redundant". Let $G'$ be the DAG obtained after deleting the node $\lambda$. Then,  $\mathcal{C}_G=\mathcal{C}_{G'}$.
\end{lemma}

Like Lemma~\ref{lemma_exogenization}, Lemma~\ref{lemma_no_redundancy}  also shows conditions under which proving the \;\;\interestingness of one DAG automatically gives you the \interestingness of another.

For example, Lemmas~\ref{lemma_exogenization} and~\ref{lemma_no_redundancy}  show that all the three DAGs of Figure~\ref{fig3} have the same sets $\mathcal{C}$ and $\mathcal{I}$. Thus, we need only to examine one DAG out of the these. It makes sense to pick the DAG of Figure~\ref{fig3}(c), as it is the simpler.

Following this same idea, we can work with the concept of an mDAG, first defined in Ref.~\cite{Evans2015}. The definition below is different than, but equivalent to the one presented in Ref.~\cite{Evans2015}.

\begin{definition}[\textbf{mDAG}]
	An mDAG is a DAG where none of the latent nodes is redundant (as defined in Lemma~\ref{lemma_no_redundancy}) nor has any parents.
\end{definition}

For example, the DAG of Figure~\ref{fig3}(c) is an mDAG. For a fixed number of observed nodes, there is a finite number of mDAGs. In particular, for 3 observed nodes there are 46 mDAGs, while for 4 observed nodes there are 2809 mDAGs. Lemmas~\ref{lemma_exogenization} and~\ref{lemma_no_redundancy} show that the mDAG encodes all the necessary information of the DAG if you only want to talk about the sets $\mathcal{C}$ and $\mathcal{I}$.

\begin{flushleft}
	\begin{figure}[h!]
		\hspace{-2mm}\begin{tabular}{ccc}
			\begin{tikzpicture}[scale=0.8]
				\node[c](A) at (0,0){$A$};
				\node[q](C) at (4,0){$C$};
				\node[q](B) at (2,0){$B$}
				edge[e] (A)
				edge[e] (C);
				\node[c](D) at (0,2){$D$}
				edge[e] (B);
				\node[c](E) at (4,2){$E$}
				edge[e] (B);
				\node[c](F) at (2,-2){$F$}
				edge[e] (C);
			\end{tikzpicture} $\longrightarrow$
			&
			\begin{tikzpicture}[scale=0.8]
				\node[c](A) at (0,0){$A$};
				\node[q](C) at (4,0){$C$};
				\node[q](B) at (2,0){$B$};
				\node[c](D) at (0,2){$D$}
				edge[e] (A)
				edge[e] (C)
				edge[e] (B);
				\node[c](E) at (4,2){$E$}
				edge[e] (A)
				edge[e] (C)
				edge[e] (B);
				\node[c](F) at (2,-2){$F$}
				edge[e] (C);
			\end{tikzpicture} $\longrightarrow$
			&
			\begin{tikzpicture}[scale=0.8]
				\node[c](A) at (0,0){$A$};
				\node[q](C) at (4,0){$C$};
				\node[c](D) at (0,2){$D$}
				edge[e] (A)
				edge[e] (C);
				\node[c](E) at (4,2){$E$}
				edge[e] (A)
				edge[e] (C);
				\node[c](F) at (2,-2){$F$}
				edge[e] (C);
			\end{tikzpicture}  \\ { } &{ } & { } \\
			(a) & (b) & (c)
		\end{tabular}
		\caption{The DAGs (a) and (b) are classically observationally equivalent to the mDAG of (c). The step (a)$\rightarrow$(b) can be shown by Lemma~\ref{lemma_exogenization}, and the step (b)$\rightarrow$(c) can be shown by Lemma~\ref{lemma_no_redundancy}.}
		\label{fig3}
	\end{figure}
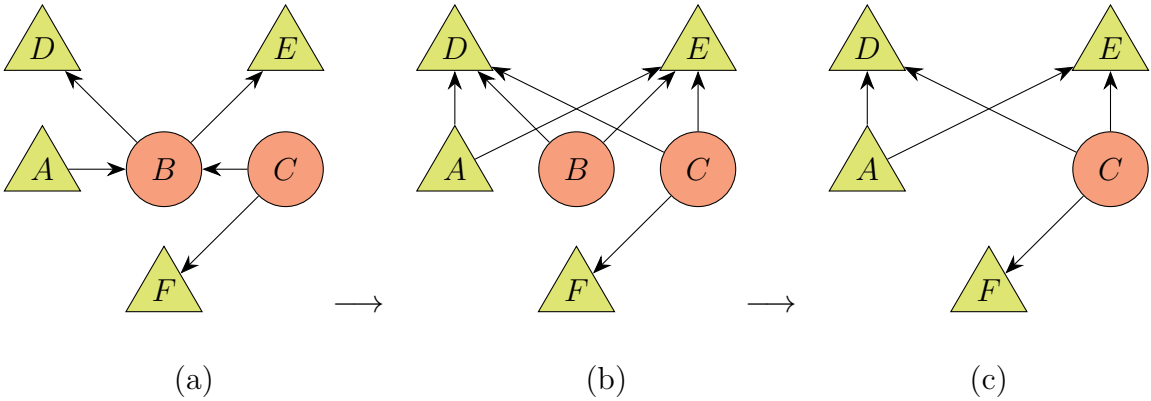
\end{flushleft}

The Lemmas here presented also give an argument in favour of counting the causal structures in terms of the number of observed nodes instead of in terms of the total number of nodes; Lemma~\ref{lemma_no_redundancy} shows that DAGs~\ref{fig3}(a) and~\ref{fig3}(b), that have 6 total nodes, actually do not have to be analyzed; their \interestingness can be examined by looking at~\ref{fig3}(c), that has 5 total nodes.

\subsection{The HLP criterion}
\label{sec_HLP_criterion}

In this section, we describe the sufficient criterion for \noninterestingness that was developed in \HLP. The criterion gives some transformations that take a DAG $G$ to another DAG $H$ such that $\mathcal{C}_H\subseteq \mathcal{C}_G$ while $\mathcal{I}_H=\mathcal{I}_G$. If the final DAG $H$ is known to be \noninteresting (for example, by being latent-free), then the original DAG $G$ is also \noninteresting.

These transformations, adjusted to the language of mDAGs, are presented in Theorem~\ref{hlp_theorem}.

\begin{theorem}
	Let $G$ and $H$ be two mDAGs. Suppose that $H$ can be obtained by starting from $G$ and applying one or more of the following transformations:
	\begin{enumerate}
		\item Removal of an edge.
		\item Addition of an edge ${X \rightarrow Y}$, where previously ${PA(X) \subseteq PA(Y)}$ and $PA(X)$ contained at least one latent node.
	\end{enumerate}

	Then, ${\mathcal{C}_{H} \subseteq \mathcal{C}_G}$.
	\label{hlp_theorem}
\end{theorem}

Now we state the \emph{HLP criterion} as a corollary:

\begin{corollary}[HLP Criterion]
	\label{hlp_criterion}
	Let $G$ be an mDAG. Suppose that by a sequence of the transformations defined in Theorem~\ref{hlp_theorem} it is possible to start from $G$ and reach another mDAG $H$ such that:
	\begin{enumerate}
		\item $H$ does not have latent nodes.
		\item The set of observed \(d\)-separation relations of $H$ and $G$ is the same, i.e. $\mathcal{I}_H=\mathcal{I}_G$.
	\end{enumerate}

	Then, the original mDAG $G$ is \noninteresting.
\end{corollary}
\begin{proof}
	From Theorem~\ref{hlp_theorem}, $\mathcal{C}_H \subseteq \mathcal{C}_G$. Since $H$ is latent-free, by Theorem~\ref{th_latent_free} it is \noninteresting, i.e. $\mathcal{C}_H=\mathcal{I}_H$. Therefore, we have ${\mathcal{I}_G =\mathcal{I}_H=\mathcal{C}_H\subseteq \mathcal{C}_G}$. As noted before, ${\mathcal{C}\subseteq\mathcal{I}}$ is valid for any DAG, due to Theorem~\ref{th_dsep_latent_permitting}. Therefore, ${\mathcal{C}_G=\mathcal{I}_G}$, meaning that $G$ is \noninteresting.
\end{proof}

Figure~\ref{hlpfig}(a) exemplifies an mDAG that can be shown \noninteresting by the \HLP Criterion: by a sequence of transformations defined in Theorem~\ref{hlp_theorem} we can obtain the mDAG shown in Figure~\ref{hlpfig}(c), that obeys the conditions presented in Corollary~\ref{hlp_criterion}.

\begin{flushleft}
	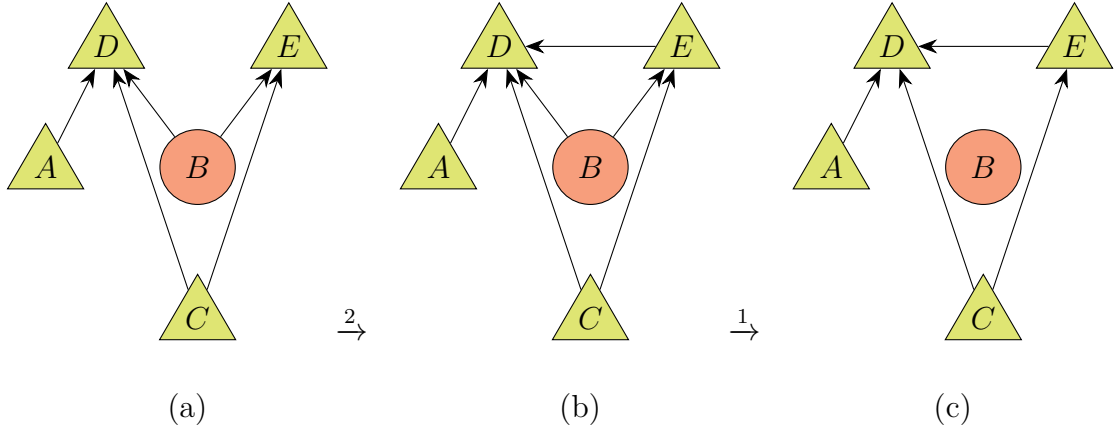
\begin{figure}[h!]
		\begin{tabular}{ccc}
			\begin{tikzpicture}[scale=0.8]
				\node[c](A) at (0,0){$A$};
				\node[c](C) at (2.5,-2.5){$C$};
				\node[q](B) at (2.5,0){$B$};
				\node[c](D) at (1,2){$D$}
				edge[e](A)
				edge[e] (C)
				edge[e] (B);
				\node[c](E) at (4,2){$E$}
				edge[e] (C)
				edge[e] (B);
			\end{tikzpicture} $\xrightarrow{2}$
			&
			\begin{tikzpicture}[scale=0.8]
				\node[c](A) at (0,0){$A$};
				\node[c](C) at (2.5,-2.5){$C$};
				\node[q](B) at (2.5,0){$B$};
				\node[c](D) at (1,2){$D$}
				edge[e](A)
				edge[e](C)
				edge[e](E)
				edge[e] (B);
				\node[c](E) at (4,2){$E$}
				edge[e](C)
				edge[e](B);
			\end{tikzpicture} $\xrightarrow{1}$
			&
			\begin{tikzpicture}[scale=0.8]
				\node[c](A) at (0,0){$A$};
				\node[q](B) at (2.5,0){$B$};
				\node[c](C) at (2.5,-2.5){$C$};
				edge[e] (C);
				\node[c](D) at (1,2){$D$}
				edge[e](A)
				edge[e](E)
				edge[e](C);
				\node[c](E) at (4,2){$E$}
				edge[e](C);
			\end{tikzpicture}  \\ { } &{ } & { } \\
			(a) & (b) & (c)
		\end{tabular}
		\caption{An example of the application of the HLP criterion. All of these three mDAGs have the same set of \(d\)-separation relations: $A\perp E$, $A\perp E|C$,$A\perp C$ and $A\perp C|D$. Since (c) is latent-free, we can conclude that (a) is \noninteresting.}
		\label{hlpfig}
	\end{figure}
\end{flushleft}

The natural question to ask here is whether the \HLP criterion is also necessary for an mDAG to be \noninteresting. The conjecture that the \HLP criterion is also necessary will be called the \emph{HLP conjecture}:

\begin{conjecture}[HLP Conjecture]
	\label{hlp_conjecture1}
	Let $G$ be an \noninteresting mDAG. Then, by a sequence of transformations defined in Theorem~\ref{hlp_theorem}, it is possible to start from $G$ and reach another mDAG $H$ such that:
	\begin{enumerate}
		\item $H$ does not have latent nodes.
		\item The set of observed \(d\)-separation relations of $H$ and $G$ is the same, i.e. $\mathcal{I}_H=\mathcal{I}_G$.
	\end{enumerate}
\end{conjecture}

As we will see, our current results do not prove this conjecture, but give hints towards its validity.

In Ref.~\cite{Evans_2022}, Evans has shown that:

\begin{theorem}
	\label{equivalence_evans}
	A latent-permitting DAG $G$ is \noninteresting if and only if it is classically observationally equivalent to a latent-free DAG $H$.
\end{theorem}

This result allows us to restate the \HLP conjecture in a different manner:

\begin{conjecture}[Reformulation of the HLP conjecture]
	Let $G$ be an mDAG. $G$ is classically observationally equivalent to a latent-free mDAG $H$ if and only if:
	\begin{enumerate}
		\item The set of observed \(d\)-separation relations of $H$ and $G$ is the same, i.e. $\mathcal{I}_H=\mathcal{I}_G$.
		\item Through the transformations defined in Theorem~\ref{hlp_theorem}, it is possible to go from $G$ to some latent-free mDAG $H'$ which has the same set of observed \(d\)-separation relations as $G$ and $H$, i.e. $\mathcal{I}_{H'}=\mathcal{I}_H=\mathcal{I}_G$.
	\end{enumerate}
\end{conjecture}

Therefore, proving the \HLP conjecture would also be of relevance to the problem of classifying causal structures into classical observational equivalence classes.

\section{Methods to determine \interestingness}
\label{methods_for_interestingness}

In this section we discuss the methods we used to prove the \interestingness of a large number of DAGs that do not respect the \HLP criterion.

\subsection{Using nonmaximality to prove \interestingness}
\label{sec_nonmaximal}

The first method to show \noninterestingness that we will present relies on the concept of \emph{maximality}~\cite{Richardson2002, Evans_2022}. To define this, we will first define what are \emph{adjacent} and \emph{\(d\)-separable} pairs of nodes.

\begin{definition}[\textbf{Adjacency}]
	Let $G$ be an mDAG, and let $A$ and $B$ be a pair of nodes observed nodes of $G$. We say that $A$ and $B$ are \emph{adjacent} in $G$ if $A$ is a parent of $B$, or $B$ is a parent of $A$, or $A$ and $B$ share a common latent parent.\footnote{For DAGs which are not exogenized (see Definition~\ref{exg}) we additionally say that $A$ and $B$ are adjacent if there is some \emph{latent treck} from on to the other. Which is to say, $A$ and $B$ are said to be \emph{adjacent} in a non-exogenized DAG $G$ iff they are adjacent in the mDAG resulting from exogenizing $G$ per see Definition~\ref{exg}.}
\end{definition}

\begin{definition}[\textbf{\(\boldsymbol{d}\)-(un)separable pair of observed nodes}]
	Let $G$ be a DAG with nodes ${\boldsymbol{A}=\boldsymbol{V}\cup \boldsymbol{L}}$, where $\boldsymbol{V}$ are observed nodes and $\boldsymbol{L}$ are latent nodes. A pair of observed nodes $\{A,B\}$ for $A\in \boldsymbol{V}$, $B\in\boldsymbol{V}$ is said to be \emph{\(d\)-separable} if there is some subset $\boldsymbol{Z}\subseteq \boldsymbol{V}$ of the remaining observed nodes such that $(A\perp B|\boldsymbol{Z})$, otherwise, $A$ and $B$ are said to be \(d\)-unseparable.
\end{definition}

These two definitions are related, in that they are criteria with relate to the compatibility of a particular distribution, namely, perfect correlation between one pair of nodes while all other nodes are point distributed. Consider the following distribution:
\begin{equation}
    P(A,B,C,D...)= \frac{1}{2}(\delta_{A,0}\delta_{B,0} + \delta_{A,1}\delta_{B,1})\delta_{C,0}\delta_{D,0}...
    \label{explicit_distribution_esep}
\end{equation}

Eq.~\eqref{explicit_distribution_esep} describes a probability distribution in which $A$ and $B$ are random but perfectly correlated, and every other observed node is point-distributed at the value $0$. Then,

\begin{proposition}[Adjacency $\Leftrightarrow P_{\text{\eqref{explicit_distribution_esep}}}\in \mathcal{C}_G$]\label{prop:pairwise_adjacency}
    The distribution in Eq.~\eqref{explicit_distribution_esep} is classically compatible with the graph $G$ if and only if $A$ and $B$ are adjacent nodes in $G$.
\end{proposition}
\begin{proof}See Refs.~\cite{Evans2015, Evans2012}. Clearly $P_{\text{\eqref{explicit_distribution_esep}}}\in \mathcal{C}_G$ when $A$ and $B$ \emph{are} adjacent. Whenever $A$ and $B$ are \emph{not} adjacent, then evidently $A$ and $B$ are \(e\)-separated\footnote{The concept of \(e\)-separation will be elaborated upon in Section~\ref{subsec:esep} here.} upon removal of $\boldsymbol{V}\setminus \{A,B\}$, which also means that $P_{\text{\eqref{explicit_distribution_esep}}}$ would violate the entropic inequality in Theorem 5 of Ref.~\cite{finkelstein2021entropic}.
\end{proof}
\begin{proposition}[\(d\)-unseparability $\Leftrightarrow P_{\text{\eqref{explicit_distribution_esep}}}\in \mathcal{I}_G$]\label{prop:pairwise_d_unseparability}
    The distribution in Eq.~\eqref{explicit_distribution_esep} satisfies all the conditional independence constraints that follow from the observed \(d\)-separation relations relations of graph $G$ if and only if $A$ and $B$ are \(d\)-unseparable in $G$.
\end{proposition}
\begin{proof}
    The only conditional independence relations that the distribution in Eq.~\eqref{explicit_distribution_esep} \emph{fails} to satisfy are those of the form $A\dbot B|\boldsymbol{Z}$. Such conditional independence relations follow from the observed \(d\)-separation relations relations of graph $G$ if and only if $A$ and $B$ are \(d\)-separable in $G$.
\end{proof}

Putting Propositions~\ref{prop:pairwise_adjacency} and~\ref{prop:pairwise_d_unseparability} together, we find that $P_{\text{\eqref{explicit_distribution_esep}}}\in \mathcal{I}_G$ but $P_{\text{\eqref{explicit_distribution_esep}}}\not\in \mathcal{C}_G$ whenever a graph has $A$ and $B$ nonadjacent but also \(d\)-unseparable. This leads us to the concept of \emph{maximality}~\cite{Richardson2002, Evans_2022}:

\begin{definition}[\textbf{Maximal DAG}]
	\label{def_maximal}
	Let $G$ be a DAG. If all of the pairs of nodes of $G$ which are \emph{not} \(d\)-separable are also adjacent then $G$ is said to be \emph{maximal}, otherwise $G$ is said to be \emph{nonmaximal}.
\end{definition}

Figure~\ref{fig_maximal} shows an example of a maximal DAG (\ref{fig_maximal}a) and an example of a nonmaximal DAG (\ref{fig_maximal}b).

\begin{flushleft}
	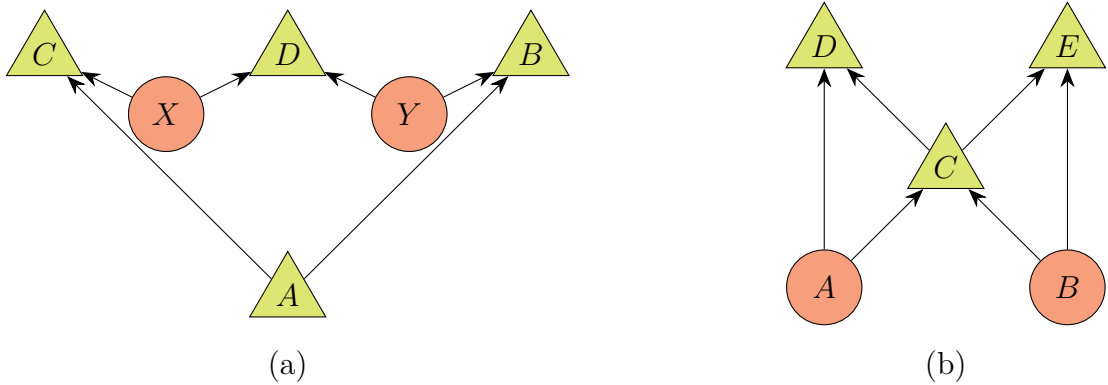
\begin{figure}[h!]
		\centering
		\begin{tabular}{ccc}
			\makecell{\begin{tikzpicture}[scale=0.8]
					\node[c](F) at (0,){$A$};
					\node[q](A) at (-2,4){$X$};
					\node[q](B) at (2,4){$Y$};
					\node[c](C) at (0,5){$D$}
					edge[e](A)
					edge[e](B);
					\node[c](D) at (-4,5){$C$}
					edge[e](A)
					edge[e](F);
					\node[c](E) at (4,5){$B$}
					edge[e](B)
					edge[e](F);
				\end{tikzpicture}}
			    & \hspace{2cm} &
			\makecell{\begin{tikzpicture}[scale=0.8]
					\node[q](A) at (0,0){$A$};
					\node[q](B) at (4,0){$B$};
					\node[c](C) at (2,2){$C$}
					edge[e] (A)
					edge[e] (B);
					\node[c](E) at (4,4){$E$}
					edge[e] (B)
					edge[e] (C);
					\node[c](D) at (0,4){$D$}
					edge[e] (A)
					edge[e] (C);
				\end{tikzpicture}}
			\\
			(a) & \hfill       & (b)
		\end{tabular}
		\caption[.]{(a) An example of maximal DAG, which so happens to be \interesting. (b) An example of an nonmaximal DAG, namely, The Unrelated Confounders scenario introduced in Ref.~\cite{Evans2012}. It is not maximal because nodes $D$ and $E$ are not \(d\)-separable, but are nevertheless not adjacent. Like \emph{all} nonmaximal DAGs, the Unrelated Confounders scenario in \interesting.}
		\label{fig_maximal}
	\end{figure}
\end{flushleft}

\settheoremtag{Nonmaximal}
\begin{theorem}\label{thm:elie_esep}
	Every \noninteresting DAG is maximal, that is, every nonmaximal DAG is \interesting.
\end{theorem}
\begin{proof}
	As we have seen from Propositions~\ref{prop:pairwise_adjacency} and~\ref{prop:pairwise_d_unseparability}, if a graph $G$ is nonmaximal, then there is some pair of observed nodes such that the distribution given by those-two-nodes-are-random-and-perfectly-correlated-while-all-other-observed-nodes-are-point-distributed lies in some \emph{gap} between $\mathcal{I}_G$ and $\mathcal{C}_G$.
	Alternatively, note that nonmaximality implies \interestingness also follows from the fact that all latent-free graphs are maximal~\cite[Prop. 3.19]{Richardson2002}. We then simply note that a nonmaximal graph cannot be observationally equivalent to any maximal graph: It follows from Eq.~\ref{eq_different_I} that every pair of observationally equivalent DAGs must agree on their sets of \(d\)-(un)separable observed-node pairs. Furthermore, pursuant to Lemma~\ref{lem:skeleton_equivalence} (discussed in Appendix~\ref{app:skeleton}), agreement with respect to adjacency structure is \emph{also} a prerequisite for observational equivalence~\cite{Evans2015, Evans2012}. These facts imply that if a DAG $G$ is nonmaximal, then it is \emph{not} going to be observationally equivalent to any latent-free DAG (and will thus be \interesting by Theorem~\ref{equivalence_evans}).
\end{proof}

Figure~\ref{fig8} shows an example of a 4-observed-nodes mDAG whose \interestingness may be shown by Theorem~\ref{thm:elie_esep}.

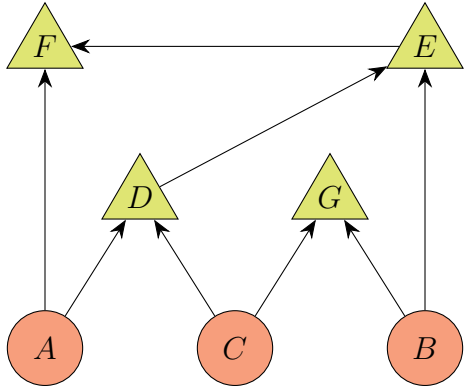
\begin{figure}
	\centering
	\begin{tikzpicture}
		\node[q](A) at (0,0){$A$};
		\node[q](B) at (5,0){$B$};
		\node[q](C) at (2.5,0){$C$};
		\node[c](D) at (1.25,2){$D$}
		edge[e] (A)
		edge[e] (C);
		\node[c](G) at (3.75,2){$G$}
		edge[e] (B)
		edge[e] (C);
		\node[c](E) at (5,4){$E$}
		edge[e] (B)
		edge[e] (D);
		\node[c](F) at (0,4){$F$}
		edge[e] (A)
		edge[e] (E);
	\end{tikzpicture}
	\caption{A DAG with 4 observed nodes that is shown to be \interesting per Theorem~\ref{thm:elie_esep}. This can be seen because $G$ and $F$ not \(d\)-separable, i.e. none of the \(d\)-separation relations $(G\perp F|E)$, $(G\perp F|D)$ or $(G\perp F|D,E)$ hold, but they are not adjacent.}
	\label{fig8}
\end{figure}

To find out which sets of pairs of observed nodes are \(d\)-separable in $G$ and then check whether it is maximal, is it necessary to first obtain \emph{all} the \(d\)-separation relations of $G$? The following accessory lemma shows that this is not the case, thus simplifying the application of Theorem~\ref{thm:elie_esep} in practice.

\begin{lemma}[Rapid test for \(d\)-separability of pairs]
	\label{lemma_dseparability}
	A pair of node subsets $\boldsymbol{A}$ and $\boldsymbol{B}$ of a DAG $G$ are \(d\)-separable by \emph{some} set of observed nodes if and only if they are \(d\)-separable by a \emph{particular} set of observed nodes, namely, the set of all and only those observed nodes which are ancestors of either $\boldsymbol{A}$ or of $\boldsymbol{B}$. In other words, ${\boldsymbol{A} \perp \boldsymbol{B} \mid \operatorname{ANC_G}(\boldsymbol{A})\cup\operatorname{ANC_G}(\boldsymbol{B})}$ whenever there exists some set $\boldsymbol{Z}$ such that ${\boldsymbol{A} \perp \boldsymbol{B} | \boldsymbol{Z}}$.
\end{lemma}
\begin{proof}
	This follows from Theorem 6 of Ref.~\cite{tian1998finding}. In particular, it follows from the use of Ref.~\cite{tian1998finding}'s Algorithm 3 to solve Ref.~\cite{tian1998finding}'s Problem 4.
\end{proof}

In their Appendix E, \HLP effectively utilized Shannon-type inequalities to certify the \interestingness of many DAGs with 5 or 6 total nodes. In particular, they found that the Shannon-type inequalities for those DAGs could by violated by distributions where two particular variables are perfectly correlated and all other observable variables are point distributed, identical in nature to the construction of Eq~\eqref{explicit_distribution_esep}. For each of these examples, \HLP highlighted the two nodes corresponding to the perfectly correlated variables in yellow. We observe here that the node pairs highlighted by \HLP are precisely pairs of nodes which are not \(d\)-separable but also not adjacent. Accordingly, our Theorem \ref{thm:elie_esep} similarly certifies the \interestingness of all those examples.

\subsection{Using setwise nonmaximality to prove \interestingness}

In the previous subsection we found that that membership in $\mathcal{I}_G$ or $\mathcal{C}_G$ of the distribution in which \emph{pair} of observed nodes is perfectly correlated while all other observed nodes are point distributed is directly related to the concepts of \(d\)-unseparability and adjacency, respectively. Here we formulate analogous criteria in order to assess perfect correlation of three-or-more variables when all other variables in the distribution are point-distributed.

\begin{definition}[\textbf{Setwise adjacency}]
    Let $G$ be a DAG with nodes ${\boldsymbol{A}=\boldsymbol{V}\cup \boldsymbol{L}}$, where $\boldsymbol{V}$ are observed nodes and $\boldsymbol{L}$ are latent nodes. Then, the subset of observed nodes $\{V_1...V_k\}$ is setwise adjacent in $G$ if and only if there is some node $X$ (possibly but not necessarily within  $\{V_1...V_k\}$) such that $X$ is an ancestor of \emph{every} node in $\{V_1...V_k\}$ not only in the DAG $G$ but also in the \emph{subgraph} of $G$ formed by deleting all nodes ${\boldsymbol{V}\setminus\{V_1...V_k\}}$ from $G$.\footnote{Note that we are using the convention that a node counts as its own ancestor, a la Refs.~\cite{Richardson1997, Steudel2015, vanderZander2019}.}
\end{definition}
\begin{definition}[\textbf{Setwise \(d\)-unrestriction}]
    Let $\boldsymbol{V}$ be the set of all observed nodes in some DAG $G$, and let $\boldsymbol{S}$ be some subset of $\boldsymbol{V}$. Then, the nodes $\boldsymbol{S}$ are  setwise \(d\)-unrestricted in $G$ if and only if there does \emph{not} exist any pair of nodes $\{S_i, S_j\}\subset \boldsymbol{S}$ along with some (possibly empty) set of observed nodes $\boldsymbol{Z}\subset \boldsymbol{V}\setminus\boldsymbol{S}$ such that $(S_i\perp S_j|\boldsymbol{Z})$.
\end{definition}

We next show that these definitions for setwise adjacency and setwise \(d\)-unrestriction have the desired properties. Consider the following distribution:
\begin{equation}
    P(\boldsymbol{V})= \frac{1}{2}(\delta_{V_1,0}\delta_{V_2,0}...\delta_{V_k,0} + \delta_{V_1,1}\delta_{V_2,1}...\delta_{V_k,1})\delta_{V_{k+1},0}\delta_{V_{k+2},0}...
    \label{setwise_correlation}
\end{equation}
Eq.~\eqref{setwise_correlation} describes a probability distribution in which the first $k$ observed variables are random but perfectly correlated while all other observed variables are point-distributed at the value $0$. Then,
\begin{proposition}[Setwise Adjacency $\Leftrightarrow P_{\text{\eqref{setwise_correlation}}}\in \mathcal{C}_G$]\label{prop:setwise_adjacency}
    The distribution in Eq.~\eqref{setwise_correlation} is classically compatible with graph $G$ if and only if $\{V_1...V_k\}$ are setwise adjacent in $G$.
\end{proposition}
 Proposition~\ref{prop:setwise_adjacency} follows from the two accessory lemmas below.
\begin{lemma}[Partial point distribution $\implies$ subgraph compatibility]\label{lem:point_distribution_subgraph}
    Suppose $P(\boldsymbol{V})$ is some distribution wherein the variables ${\boldsymbol{V}\setminus\{V_1...V_k\}}$ are point-distributed and moreover all the variables in $\{V_1...V_k\}$ have finite cardinality. 
    Then, $P(\boldsymbol{V})$ is classically compatible with $G$ if and only if $P(\{V_1...V_k\})$ is compatible with the \emph{subgraph} of $G$ formed be deleting all nodes ${\boldsymbol{V}\setminus\{V_1...V_k\}}$ from $G$. 
\end{lemma}
\begin{proof}
    This lemma is an immediate consequence of the \(e\)-separation theorem central in Ref.~\cite{Evans2012}.
\end{proof}

\begin{lemma}[setwise correlation $\implies$ common ancestor]
    Let $P_{\textrm{perfect correlation}}(\boldsymbol{A})$ be the distribution in which all variables in $\boldsymbol{A}$ are random but perfectly correlated with each other. %
    Then, $P_{\textrm{perfect correlation}}(\boldsymbol{A})$ is compatible with a DAG $G$ if and only if all the nodes in $\boldsymbol{A}$ share some common ancestor in $G$.
\end{lemma}
\begin{proof}
    The \enquote{if} direction is trivial; the \enquote{only if} direction follows from Ref.~\cite[Theorem~2]{Steudel2015}.
\end{proof}

Note that Proposition~\ref{prop:setwise_adjacency} implies Proposition~\ref{prop:pairwise_adjacency} as a special case: A pair of observed nodes share a common ancestor upon removing all other observed nodes from a DAG $G$ if and and only they are adjacent in $G$.

We likewise highlight the utility of the definition of a setwise \(d\)-unrestricted set:
\begin{proposition}[Setwise \(d\)-unrestriction $\Leftrightarrow P_{\text{\eqref{setwise_correlation}}}\in \mathcal{I}_G$]\label{prop:setwise_d_unseparability}
    The distribution in Eq.~\eqref{setwise_correlation} satisfies all the conditional independence constraints that follow from the observed \(d\)-separation relations relations of graph $G$ if and only if $\{V_1...V_k\}$ are setwise \(d\)-unrestricted in $G$.
\end{proposition}
\begin{proof} Suppose that $\{V_1...V_k\}$ are \emph{not} setwise \(d\)-unrestricted in $G$. Then, there exists a pair of nodes $\{S_i,S_j\}\subseteq\{V_1...V_k\}$ such that $G$ exhibits the \(d\)-separation relation $(S_i\perp S_j|\boldsymbol{Z})$. In this case, the distribution
    \begin{align}\label{eq:setwise_correlation_S}
        P(\boldsymbol{V})= \frac{1}{2}\left(\delta_{\{V_1...V_k\},\vec{0}^{k}} + \delta_{\{V_1...V_k\},\vec{1}^{k}}\right)\delta_{\boldsymbol{V}\setminus\{V_1...V_k\},\vec{0}^{(|\boldsymbol{V}|-k)}}\footnotemark
    \end{align}
    \footnotetext{Here $|\mathbf{X}|$ denotes the number of elements of the set $\mathbf{X}$, and $\vec 0^{n}$ represents the vector with $n$ entries equalling $0$.}
    violates the conditional independence relation $S_i\dbot S_j|\boldsymbol{Z}$, and therefore $P_{\text{\eqref{setwise_correlation}}}\notin \mathcal{I}_G$. For the other direction, note that if $(S_i\not\perp S_j|\boldsymbol{Z})$ then so too $(\boldsymbol{S}_i\not\perp \boldsymbol{S}_j|\boldsymbol{Z})$ for any disjoint sets $\boldsymbol{S}_i$ and $\boldsymbol{S}_j$ wherein ${S_i \in \boldsymbol{S}_i}$ and ${S_j \in \boldsymbol{S}_j}$. Consequently, when $\{V_1...V_k\}$ are setwise \(d\)-unrestricted in $G$, there is no way to \(d\)-separate \emph{any} subset of $\{V_1...V_k\}$ from any other subset of $\{V_1...V_k\}$ by any subset of the observed nodes outside of $\{V_1...V_k\}$, which are the \emph{only} \(d\)-separation relations whose corresponding conditional independence relations would exclude the distribution of Eq.~\eqref{eq:setwise_correlation_S}. This means that, in this case, $P_{\text{\eqref{setwise_correlation}}}\in \mathcal{I}_G$.
\end{proof}

Putting Propositions~\ref{prop:setwise_adjacency} and~\ref{prop:setwise_d_unseparability} together lead us to a natural generalization of maximality, which we now define and employ in a theorem.
\begin{definition}[\textbf{Setwise Maximal DAG}]
    \label{def_setwise_maximal}
    Let $G$ be DAG. If every subset of the observed nodes of $G$ which is setwise \(d\)-unrestricted is also setwise adjacent then $G$ is said to be \emph{setwise-maximal}, otherwise $G$ is said to be \emph{setwise nonmaximal}.
\end{definition}
\settheoremtag{Setwise Nonmaximal}
\begin{theorem}\label{thm_setwise_nonmaximal}
    Every \noninteresting DAG is setwise-maximal, that is, every setwise-nonmaximal DAG is \interesting.
\end{theorem}
 \begin{proof}
      Theorem~\ref{thm_setwise_nonmaximal} follows immediately from Propositions~\ref{prop:setwise_adjacency} and~\ref{prop:setwise_d_unseparability}.
 \end{proof}

\begin{flushleft}
	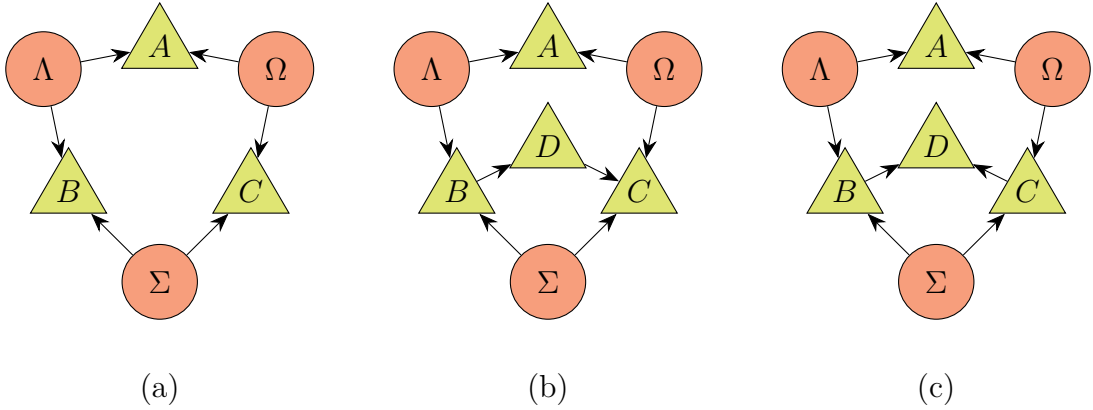
\begin{figure}[h!]
		\centering
		\begin{tabular}{ccccc}
			\begin{tikzpicture}[scale=0.6]
				\node[q](L) at (-0.55,2.7){$\Lambda$};
				\node[q](Z) at (2,-2){$\Sigma$};
				\node[q](O) at (4.55,2.7){$\Omega$};
				\node[c](A) at (0,0){$B$}
				edge[e] (L)
				edge[e] (Z);
				\node[c](C) at (4,0){$C$}
				edge[e] (O)
				edge[e] (Z);
				\node[c](D) at (2,3.2){$A$}
				edge[e] (O)
				edge[e] (L);
			\end{tikzpicture}
			& \hspace{0.2cm} &
   			\begin{tikzpicture}[scale=0.6]
				\node[q](L) at (-0.55,2.7)
				{$\Lambda$};
				\node[q](Z) at (2,-2){$\Sigma$};
				\node[q](O) at (4.55,2.7){$\Omega$};
				\node[c](A) at (0,0){$B$}
				edge[e] (L)
				edge[e] (Z);
				\node[c](D) at (2,1){$D$}
				edge[e](A);
    			\node[c](C) at (4,0){$C$}
				edge[e] (O)
				edge[e] (Z)
    			edge[e] (D);
				\node[c](X) at (2,3.2){$A$}
				edge[e] (O)
				edge[e] (L);
			\end{tikzpicture}
			& \hspace{0.2cm} &
			\begin{tikzpicture}[scale=0.6]
				\node[q](L) at (-0.55,2.7)
				{$\Lambda$};
				\node[q](Z) at (2,-2){$\Sigma$};
				\node[q](O) at (4.55,2.7){$\Omega$};
				\node[c](A) at (0,0){$B$}
				edge[e] (L)
				edge[e] (Z);
				\node[c](C) at (4,0){$C$}
				edge[e] (O)
				edge[e] (Z);
				\node[c](D) at (2,1){$D$}
                    edge[e](C)
				edge[e](A);
				\node[c](X) at (2,3.2){$A$}
				edge[e] (O)
				edge[e] (L);
			\end{tikzpicture}
			\\ { } &{ } & { }&{ } & { } \\
			(a) & { } & (b) & { } & (c)
		\end{tabular}
		\caption{Examples of maximal DAGs which are setwise-nonmaximal. (a) is the Triangle scenario, where the set $\{A,B,C\}$ (all visible nodes) is setwise \(d\)-unrestricted but not setwise adjacent. In (b), the set $\{A,B,C\}$ is setwise \(d\)-unrestricted (in fact, there are no \(d\)-separation relations between observed nodes) but not setwise adjacent; by contrast, the larger set $\{A,B,C,D\}$ is \emph{both} setwise \(d\)-unrestricted \emph{and} setwise adjacent, despite (b) exhibiting the \(d\)-separation relation ${(A\perp D|B)}$, as that \(d\)-separation relation involves conditioning on a node \emph{within the set}. In (c), both of the sets $\{A,B,C\}$ and $\{A,B,C,D\}$ are \(d\)-unrestricted but not setwise adjacent. Note that $\{A,B,C,D\}$ is \(d\)-unrestricted in the DAG (c) despite that DAG exhibiting the \(d\)-separation relation ${(A\perp D|B,C)}$, as that \(d\)-separation relation involves conditioning on nodes \emph{within the set}.}
		\label{fig_setwise_nonmaximal}
	\end{figure}
\end{flushleft}

The DAGs of Fig.~\ref{fig_setwise_nonmaximal} are setwise-nonmaximal, so they are shown \interesting by Theorem~\ref{thm_setwise_nonmaximal} (even if they are \emph{not} shown \interesting by Theorem~\ref{thm:elie_esep}). On the other hand, the mDAG of Fig.~\ref{fig_maximal}(a) is both maximal and setwise maximal, as all maximal DAGs are also setwise maximal.

Note that there are precisely five mDAGs with 3 observed nodes which are \interesting (the Triangle scenario, the Unrelated Confounders scenario and three observationally equivalent versions of the Instrumental scenario). Theorem~\ref{thm_setwise_nonmaximal} certifies the \interestingness of \emph{all five}. Indeed, Theorem~\ref{thm_setwise_nonmaximal} turns out to be an extremely powerful filter for recognizing \interesting mDAGs with 4 or 5 observed nodes as well, as discussed in Section~\ref{sec:results}.

\subsection{Using \(d\)-separation to prove \interestingness}
\label{subsec:dsep}

With Eq.~\eqref{eq_different_I}, it was noted that any two DAGs that have different sets of observed \(d\)-separation relations are \emph{not} classically observationally equivalent. In other words, imposing the same observed conditional independence constraints on the compatible distributions is a necessary condition for two DAGs to be classically observationally equivalent.

This fact can be used to establish the \interestingness of some DAGs: if we can prove that the set of observed \(d\)-separation relations of a latent-permitting DAG \emph{does not} match the set of \(d\)-separation relations of \emph{any} latent-free DAG, then our latent-permitting DAG is classically observationally \emph{inequivalent} to all latent-free DAGs. Via Evans'~\cite{Evans_2022} Theorem~\ref{equivalence_evans}, then, we can conclude that our latent-permitting DAG is \interesting.

This type of reasoning, that says that when a certain property of a DAG $G$ is unmatched by all latent-free DAGs then $G$ is \interesting, will be used a few times in this subsection and in the two next ones. As such, it will be useful to define the auxiliary term \emph{Not Achievable in Latent-Free (\NLF)}:

\begin{definition}[\NLF property of a DAG]
	\label{def_NLF}
	Let $G$ be a latent-permitting DAG. If $G$ has a certain property that \emph{does not match} that same property of any latent-free DAG, we say that this property of $G$ is \emph{\NLF (not achievable in latent-free)}.
\end{definition}

For example, we can have a DAG $G$ whose set of observed \(d\)-separation relations is \NLF. If there is a proof that a DAG is observationally inequivalent to all latent-free DAGs whenever certain property of the DAG is \NLF, then this can be used to show \interestingness. As discussed, this is the case for \(d\)-separation.

\settheoremtag{\NLF \(d\)-sep}
\begin{theorem}\label{thm:dsep_for_interesting}
	Let $G$ be a DAG. Suppose that the set of observed \(d\)-separation relations of $G$ is \NLF (as per Definition~\ref{def_NLF}). Then, $G$ is \interesting.
\end{theorem}
\begin{proof}
	From  Eq.~\eqref{eq_different_I}, $G$ is not classically observationally equivalent to any latent-free DAG. Thus, via Theorem~\ref{equivalence_evans},  $G$ is \interesting.
\end{proof}

The mDAG  presented in Figure~\ref{fig_maximal}(a), which was not shown \interesting by Theorem~\ref{thm_setwise_nonmaximal} due to being setwise maximal, \emph{can} be shown \interesting by Theorem~\ref{thm:dsep_for_interesting}: its set of observed \(d\)-separation relations,  $A\perp D|\emptyset$ and $B\perp C|A$, is not matched by any latent-free DAG. This mDAG can be considered a special case of the bilocality scenario~\cite{bilocal} restricted to have the same setting employed at both the extreme wings. Remarkably, it can support non-classical correlations even though the \enquote{setting} for the extreme wings is the same (in stark contradiction with the Bell scenario).

On the other hand, note that the Evans scenario (Figure~\ref{fig_maximal}(b)), which was shown \interesting by Theorem~\ref{thm:elie_esep}, cannot be shown \interesting via Theorem~\ref{thm:dsep_for_interesting}: it  does not have any \(d\)-separation relation, just like the saturated latent-free DAGs. Therefore, Theorems~\ref{thm:elie_esep} and ~\ref{thm:dsep_for_interesting} are not redundant to each other. Similarly,  Theorems~\ref{thm_setwise_nonmaximal} and ~\ref{thm:dsep_for_interesting} are not redundant to each other.

There are a variety of different practical methods to ascertain whether or not Theorem~\ref{thm:dsep_for_interesting} is satisfied for a given latent-permitting DAG $G$. Naively, one could construct \emph{all} the latent-free DAGs with the same number of observed nodes as $G$, and one-by-one check whether any of them have observed \(d\)-separation relations matching those of $G$. Alternatively, one could envision employing a constraint-based causal discovery algorithm where the input to the algorithm is precisely the observed \(d\)-separation relations of $G$: If the output of the causal discovery algorithm fails to include \emph{any} latent-free DAG as a viable explanation given the input constraints, then evidently $G$ is \interesting per Theorem~\ref{thm:dsep_for_interesting}. While such \enquote{brute force} approaches are viable for a small number of observed nodes, as the number of observed nodes increases one would need to contend with potential combinatorial explosion. As it turns out, however, when the DAG is maximal (and thus not shown \interesting by Theorem~\ref{thm:elie_esep}) there is an efficient way to check whether or not its set of observed \(d\)-separation relations is \NLF; the efficient algorithm is presented in Appendix~\ref{app_rapid_dsep}.

As well as our previous methods,Theorem~\ref{thm:dsep_for_interesting} is only a \emph{sufficient} condition for \interestingness, and not necessary. If there \emph{exists} a latent-free DAG $H$ which has the same observed \(d\)-separation relations as $G$ (i.e. $\mathcal{I}_H=\mathcal{I}_G$), this \emph{does not} imply that $\mathcal{C}_G = \mathcal{C}_H$, and as such we cannot conclude anything about the relation between $\mathcal{C}_G$ and $\mathcal{I}_G$. We will now discuss other sufficient conditions for \interestingness that can be used when Theorems~\ref{thm:elie_esep} and~\ref{thm:dsep_for_interesting} fail.

\subsection{Using \textit{e}-separation to prove \interestingness}\label{subsec:esep}

As well as the mismatch of  \(d\)-separation relations is a witness of classical observational inequivalence, we can show that the mismatch of \(e\)-separation relations, a concept that will be defined below, also witnesses classical observational inequivalence. As such, we can mimic the same logic used in the last section, thus showing that \(e\)-separation relations can be used to attest \interestingness. An \(e\)-separation relation is defined as:

\begin{definition}[\textbf{\(\boldsymbol{e}\)-separation}]
	\label{esep}
	Let $G$ be a DAG, and let $\boldsymbol{X}$, $\boldsymbol{Y}$, $\boldsymbol{Z}$ and $\boldsymbol{W}$ be four disjoint sets of nodes of $G$. Let $G_{\text{del}_{\boldsymbol{W}}}$ be the DAG obtained by starting from $G$ and deleting the nodes of $\boldsymbol{W}$.
	The sets $\boldsymbol{X}$ and $\boldsymbol{Y}$ are said to be \emph{\(e\)-separated by $\boldsymbol{Z}$ upon deletion of $\boldsymbol{W}$} in $G$, denoted  $(\boldsymbol{X}\eperp \boldsymbol{Y}|\boldsymbol{Z})_{del_{\boldsymbol{W}}}$, if $\boldsymbol{X}\perp \boldsymbol{Y}|\boldsymbol{Z}$ holds in $G_{\text{del}_{\boldsymbol{W}}}$.
\end{definition}

Note that if $\boldsymbol{W}=\emptyset$, the concept of \(e\)-separation reduces to \(d\)-separation. As well as for the case of  \(d\)-separation, we will say that an  \(e\)-separation relation is an \emph{observed  \(e\)-separation relation} when the sets $\boldsymbol{X}$, $\boldsymbol{Y}$, $\boldsymbol{Z}$ and $\boldsymbol{W}$ only involve observed nodes. Matching observed \(e\)-separation relations is a prerequisite for observational equivalence:

\begin{lemma}[\(e\)-separation condition for observational equivalence]\label{lem:esepforequiv}
	Let $G$ and $H$ be two DAGs. If they are classically observationally equivalent (i.e. $\mathcal{C}_G=\mathcal{C}_H$), then their sets of observed \(e\)-separation relations must be identical.

\end{lemma}
\begin{proof}
	In Ref.~\cite{finkelstein2021entropic} it is shown that \(e\)-separation relations imply in inequalities that must be satisfied by the compatible probability distributions. In Appendix~E of that same reference, it is further shown that if a DAG \emph{does not} exhibit an \(e\)-separation relation, then there must exist a compatible probability distribution which violates the inequality associated with that \(e\)-separation relation. This implies that, if a DAG $G$ exhibits an \(e\)-separation relation which is not exhibited by another DAG $H$, then it is possible to find a probability distribution that is compatible with $H$ but not with $G$.
\end{proof}

Using this lemma, we can derive the analog of Theorem~\ref{thm:dsep_for_interesting} for the case of \(e\)-separation:

\settheoremtag{\NLF \(e\)-sep}
\begin{theorem}
	\label{new_e_sep}
	Let $G$ be a DAG. Suppose that the set of observed \(e\)-separation relations of $G$ is \NLF (as per Definition~\ref{def_NLF}). Then, $G$ is \interesting.

\end{theorem}
\begin{proof}
	Follows directly from Lemma~\ref{lem:esepforequiv} and Theorem~\ref{equivalence_evans}.
\end{proof}

To use Theorem~\ref{new_e_sep} in practice, one might enumerate the \(e\)-separation relations exhibited by every latent-free DAG with the same number of observed nodes of $G$ and compare them to the observed \(e\)-separation relations of $G$. Far more efficiently, one need only check against any \emph{one} latent-free graph which matches the observed \(d\)-separation relations of $G$. (Such a latent-free DAG must exist if $G$ is not already certified as \interesting by Theorem~\ref{thm:dsep_for_interesting}\footnote{We do not discuss \emph{how} to find a latent-free DAG with the same \(d\)-separation relations as $G$, as ultimately we conclude that there is apparently no advantage in doing so, as discussed in Appendix~\ref{app_rapid_dsep}.}). After all, if two latent-free graphs share the same \(d\)-separation relations, they will also share the same \(e\)-separation relations, per Theorem~\ref{th_latent_free} and Lemma~\ref{lem:esepforequiv}. We thus advise verifying that a DAG $G$ in question is not already certified as \interesting by Theorem~\ref{thm:dsep_for_interesting} before invoking Theorem~\ref{new_e_sep}.

It is clear that every DAG that can be shown \interesting by Theorem~\ref{thm:dsep_for_interesting} can also be shown \interesting by Theorem~\ref{new_e_sep}, since \(d\)-separation is a special case of \(e\)-separation. A little less trivial, every DAG that can be shown \interesting by Theorem~\ref{thm:elie_esep} can also be shown \interesting by Theorem~\ref{new_e_sep}: all nonmaximal DAGs have a set of \(e\)-separation relations which is \NLF.

\begin{proposition}[Theorem~\ref{new_e_sep} subsumes Theorem~\ref{thm:elie_esep}]
	\label{prop_esep_subsumes}
	Let $G$ be a nonmaximal DAG. Then, the set of \(e\)-separation relations of $G$ is \NLF.
\end{proposition}
\begin{proof}
	If $G$ is nonmaximal, then there are at least two nodes $A$ and $B$ which are not \(d\)-separable, but are also not adjacent in $G$. If they are not adjacent in $G$, it is clear that $A$ is \(e\)-separated from $B$ by deletion of \emph{every other} observed node of $G$.

	Let us make a proof by contradiction: suppose that the set of \(e\)-separation relations of $G$ \emph{is not} \NLF. Then, there is a latent-free DAG $H$ which has exactly the same set of \(e\)-separation relations as $G$. This means that $A$ should be \(e\)-separated from $B$ by deletion of every other node of $H$, which implies that $A$ and $B$ are \emph{not} adjacent in $H$. In a latent-free graph, two nodes are \emph{nonadjacent} if and only if said pair of nodes are \(d\)-separable~\cite[Prop. 3.19]{Richardson2002}. However, if $A$ and $B$ are \(d\)-separable in $H$ but not in $G$, then their sets of \(d\)-separation relations do not coincide, which is a contradiction.
\end{proof}

It is still an open question whether the conjunction of Theorems~\ref{thm:elie_esep} and~\ref{thm:dsep_for_interesting} is as good as Theorem~\ref{new_e_sep} to show \interestingness. We did not find any DAG that was shown \interesting by Theorem~\ref{new_e_sep} but not by any of these two previous methods. (See Appendix~\ref{app:5nodes} for further discussion of this question.)

Note, however, that Theorem~\ref{new_e_sep} \emph{does not} subsume Theorem~\ref{thm_setwise_nonmaximal}. For example, the Triangle scenario (Fig.~\ref{fig_setwise_nonmaximal}(a)), that is shown \interesting by  Theorem~\ref{thm_setwise_nonmaximal}, does \emph{not} have any \(e\)-separation relation (just like the saturated latent-free DAGs).

Appendices~\ref{app:pienaar} and \ref{app:skeleton} relate Theorems~\ref{thm:elie_esep} and~\ref{new_e_sep} to results of prior literature. In particular, in Appendix~\ref{app:pienaar} we show that a version of Theorem~\ref{thm:elie_esep} in terms of \(e\)-separation that was presented in Ref.~\cite{Pienaar2017} is incorrect.

\subsection{Using incompatible supports to prove \interestingness}\label{subsec:supports}

The final method we used to prove \interestingness is based on the classical feasibility of supports.

Given a set of random variables $\{X_1,...,X_n\}$, a specific set $\{X_1=x_1,...,X_n=x_n\}$ of values that these random variables can take is called an \emph{event}. The \emph{support} $\mathcal{S}(P_{X_1,...,X_n})$ of a probability distribution $P_{X_1,...,X_n}$ over the variables $\{X_1,...,X_n\}$ is the set of events that have non-zero probability:
\begin{equation}
	\mathcal{S}(P_{X_1,...,X_n})=\left\{\{x_1,...,x_n\} \: | \: P_{X_1,...,X_n}(x_1,...,x_n)> 0\right\}
\end{equation}

We previously defined what it means for a probability distribution to be classically compatible with a DAG. Here, we define what it means for a support to be classically compatible with a DAG:

\begin{definition}[\textbf{Compatibility of a support with a DAG}]\label{def:support_compatible}
	Let $G$ be a DAG with observed nodes $\boldsymbol{A}=\boldsymbol{V}\cup\boldsymbol{L}$, where $\boldsymbol{V}$ are observed nodes and $\boldsymbol{L}$ are latent nodes. Let $\mathcal{S}$ be a set of events over the variables $\boldsymbol{V}$. We say that $\mathcal{S}$ is a support \emph{classically compatible with} $G$ if there exists a probability distribution $P_{\boldsymbol{V}}$ over $\boldsymbol{V}$ that is classically compatible with $G$ (i.e., $P_{\boldsymbol{V}}\in \mathcal{C}_G$) and whose support is ${\mathcal{S}(P_{\boldsymbol{V}})=\mathcal{S}}$. We say that $\mathcal{S}$ is a support \emph{compatible-up-to-CI} with $G$ if there exists a probability distribution $P_{\boldsymbol{V}}$ over $\boldsymbol{V}$ such that $P_{\boldsymbol{V}}\in \mathcal{I}_G$) and whose support is ${\mathcal{S}(P_{\boldsymbol{V}})=\mathcal{S}}$.
\end{definition}

As an example, the following support is \emph{not} compatible with the Bell DAG (Figure~\ref{fig_Bell_DAG}), as it corresponds to the Popescu-Rohrlich box~\cite{khalfin1985steklov, popescu1994quantum}:

\begin{equation}
	\label{eq_support_Bell}
	\mathcal{S}_\text{Bell} = \left\{\begin{aligned}
		  & \{X=0,Y=0,A=0,B=0\} \\
		  & \{X=0,Y=0,A=1,B=1\} \\
		  & \{X=0,Y=1,A=0,B=0\} \\
		  & \{X=0,Y=1,A=1,B=1\} \\
		  & \{X=1,Y=0,A=0,B=0\} \\
		  & \{X=1,Y=0,A=1,B=1\} \\
		  & \{X=1,Y=1,A=1,B=0\} \\
		  & \{X=1,Y=1,A=0,B=1\}
	\end{aligned}\right\}
\end{equation}

Note that if a support \emph{is} compatible with DAG $G$, that does \emph{not} mean that every distribution with that support will be compatible with $G$. There are countless counterexamples, but let us simply note that the full support (the one where all events have positive probability) is compatible with any DAG, but at the same time we know of many incompatible distributions which nevertheless have full support.

Naturally, admitting the same set of compatible supports is a prerequisite for two DAGs to admit the same set of compatible distributions:

\begin{lemma}[Supports condition for observational equivalence]\label{lem:supportsforequiv}
	Let $G$ and $H$ be two DAGs. If they are classically observationally equivalent (i.e. $\mathcal{C}_G=\mathcal{C}_H$), then their sets of classically compatible supports must be identical.

\end{lemma}

It remains an open question whether the condition of Lemma~\ref{lem:supportsforequiv} is also necessary for observational equivalence. In particular, it is not known whether or not there exists a DAG for which some distributions are incompatible (due to inequalities) but for which all \emph{supports} are compatible.

As before, this necessary condition for observational equivalence immediately translates into a method for proving \interestingness:

\settheoremtag{\NLF Supports}
\begin{theorem}
	\label{thm:supports}
	Let $G$ be a DAG. Suppose the set of classically compatible supports of $G$ is \NLF (as per Definition~\ref{def_NLF}). Then, $G$ is \interesting.

\end{theorem}

To exploit Theorem~\ref{thm:supports} in practice, we need an  an algorithm capable of assessing whether or not a given support is compatible with a given DAG. Such algorithm was developed in Ref.~\cite{Fraser2020}, and refer to it here as \emph{Fraser's algorithm}. We have implemented Fraser's algorithm in Python and scripted it to yield \emph{all} the supports that are classically incompatible with a given DAG (for a certain assignment of the cardinalities of the observed variables).

In general, Fraser's algorithm is much more computationally expensive than simply assessing whether or not a graph exhibits some \(d\)-separation or \(e\)-separation relation. Consequently, we consider Theorem~\ref{thm:supports} a method of last resort to show \interestingness.

\subsubsection{Rapidly testing supports (without comparing to any latent-free graph)}

As well as for the case of \(e\)-separation, Theorem~\ref{thm:supports} has the downside that it requires one to find the supports compatible with the DAG $G$ and then check the compatible supports of all the latent-free DAGs with the same number of observed nodes (or alternatively to find the latent-free $H$ that has the same set of \(d\)-separation relations as $G$, and then check which supports are compatible with $H$). Since Fraser's algorithm is computationally expensive, doing this in practice can be cumbersome.

Luckily, it is possible to develop a rapid supports test where it is not even necessary to find \emph{all} of the supports compatible with $G$. The idea of the rapid supports test comes from noting that sometimes we can prove the \emph{incompatibility} of a given support with a DAG $G$ by recognizing that the given support conflicts with a \(d\)-separation relation exhibited by $G$.

Suppose, for instance, that a DAG $G$ has the (unconditional) \(d\)-separation relation $A\perp B$. Then, the support in Eq.~\eqref{demo support} is clearly incompatible with $G$, since any probability distribution with that support will must have $P_{AB}(1,1)=0\neq P_A(1)P_B(1)>0$, contradicting $A\dbot B$.
\begin{equation}
	\label{demo support}
	\mathcal{S}_\text{\ref{fig_ex_support}} = \left\{\begin{aligned}
		  & \{A=0,B=0\} \\
		  & \{A=0,B=1\} \\
		  & \{A=1,B=0\}
	\end{aligned}\right\}
\end{equation}

Indeed, we can formally categorize all such ``trivial" proofs of support incompatibility through the following two definitions:

\begin{definition}[\textbf{Support conflicting with a conditional independence relation}]
	\label{def_supp_conflict}
	Let $\mathcal{S}$ be a support over a set $\boldsymbol{V}$ of variables, and let $\boldsymbol{A}\subseteq\boldsymbol{V}$, $\boldsymbol{B}\subseteq\boldsymbol{V}$ and $\boldsymbol{C}\subseteq\boldsymbol{V}$ be three disjoint subsets of $\boldsymbol{V}$. We say that $\mathcal{S}$ \emph{conflicts with the conditional independence relation} $\boldsymbol{A}\dbot \boldsymbol{B}|\boldsymbol{C}$ if there exists a set $\{\boldsymbol{a},\boldsymbol{b},\boldsymbol{c}\}$ of values of the variables in $\boldsymbol{A}$, $\boldsymbol{B}$ and $\boldsymbol{C}$ such that the events $\{\boldsymbol{A}=\boldsymbol{a}, \boldsymbol{C}=\boldsymbol{c}\}$ and $\{\boldsymbol{B}=\boldsymbol{b}, \boldsymbol{C}=\boldsymbol{c}\}$ occur in $\mathcal{S}$, but the event $\{\boldsymbol{A}=\boldsymbol{a}, \boldsymbol{B}=\boldsymbol{b}, \boldsymbol{C}=\boldsymbol{c}\}$ does not occur in $\mathcal{S}$.
\end{definition}

For example, the support of Eq.~\eqref{demo support} conflicts with the conditional independence relation $A\dbot B$: both the events $A=1$ and $B=1$ occur in the support, but the event $\{A=1,B=1\}$ does not.

If a support $\mathcal{S}$ conflicts with a conditional independence relation $\boldsymbol{A}\dbot \boldsymbol{B}|\boldsymbol{C}$, then there is no probability distribution with support $\mathcal{S}$ that obeys $\boldsymbol{A}\dbot \boldsymbol{B}|\boldsymbol{C}$. This will be seen explicitly in the proof of Lemma~\ref{lemma_trivial_support}.

\begin{definition}[\textbf{Triviality of support incompatibility}]\label{def:support_and_dsep}
	A support $\mathcal{S}$ is said to be \emph{trivially incompatible} with a given DAG whenever the DAG exhibits some \(d\)-separation relation whose associated conditional independence relation conflicts with $\mathcal{S}$ (as in Definition~\ref{def_supp_conflict}). 
\end{definition}

By generalizing the discussion made around Eq.~\eqref{demo support}, we see that this definition indeed implies in classical incompatibility of the support with the DAG:

\begin{lemma}[Trivial incompatibility implies incompatibility]
	\label{lemma_trivial_support}
	If a support $\mathcal{S}$ is trivially incompatible with a DAG $G$, then it is classically incompatible with $G$.
\end{lemma}
\begin{proof}
	Let $\boldsymbol{A}\perp \boldsymbol{B}|\boldsymbol{C}$ be a \(d\)-separation relation of $G$ which is in conflict with $\mathcal{S}$. Furthermore, let $\{\boldsymbol{a},\boldsymbol{b},\boldsymbol{c}\}$ be a set of values of the variables in $\boldsymbol{A}$, $\boldsymbol{B}$ and $\boldsymbol{C}$ that witnesses this conflict. This means that all of the probability distributions which have the support $\mathcal{S}$ must have $P_{\boldsymbol{A} \boldsymbol{B}|\boldsymbol{C}}(\boldsymbol{a} \boldsymbol{b}|\boldsymbol{c})=0\neq P_{\boldsymbol{A} |\boldsymbol{C}}(\boldsymbol{a} |\boldsymbol{c})P_{ \boldsymbol{B}|\boldsymbol{C}}(\boldsymbol{b}|\boldsymbol{c})>0$. Therefore, $\mathcal{S}$ is not classically compatible with $G$.
\end{proof}

The key insight which allows us to accelerate the application of Theorem~\ref{thm:supports} is that for a latent-free DAG, the \emph{only} supports incompatible with it are those which are \emph{trivially} incompatible with it.
\begin{lemma}[Latent-free support compatibility]\label{lem:boring_supports}
	Let $H$ be a latent-free DAG. If $\mathcal{S}$ is a support which is not trivially incompatible with $H$ as per Definition~\ref{def:support_and_dsep}, then $\mathcal{S}$ is classically compatible with $H$ as per Definition~\ref{def:support_compatible}.\footnote{More generally, we conjecture that for \emph{any} DAG $G$, if $\mathcal{S}$ is a support which is not trivially incompatible with $G$ as per Definition~\ref{def:support_and_dsep}, then $\mathcal{S}$ is compatible \emph{up-to-CI}  with $G$. We did not prove this, however, except for the special case when $G$ is latent-free.}
\end{lemma}
\begin{proof}
	Let $X_1, \dotsc, X_n$ be the nodes of $H$ in some topological order, i.e. an order where $X_{i}$ is a non-descendant of $X_{i+1}$. We will make this proof by explicitly constructing a probability distribution $P(X_1,...,X_n)$ which is classically compatible with $H$ and has the support $\mathcal{S}$.

	The explicit construction is given by $P(X_1,...,X_n)=\prod_i P(X_i|\text{PA}_H(X_i))$, where each $P(X_i|\text{PA}_H(X_i))$ is uniformly distributed over the values of $X_i$ which occur along with the given values of $\text{PA}_H(X_i)$ (and any value of the remaining variables) in $\mathcal{S}$. It is clear that this distribution is classically compatible with $H$, since it takes the form of the Markov condition. Now, we will show that its support is in fact $\mathcal{S}$.

	We will make a proof by induction on $m$ that this distribution has the correct support on $X_1,...,X_m$. For $m=1$ we simply have $P(X_1)$ uniformly distributed on values of $X_1$ that are possible under $\mathcal{S}$, so the basis case is immediately satisfied. Now, we assume that $X_1,...,X_k$ has the correct support, and we will prove that $X_1,...,X_{k+1}$ also does.

	\sloppy First, consider some $(x_1, \dotsc, x_k, x_{k+1})$ that occurs in $\mathcal{S}$. Since we assumed that $X_1,...,X_k$ has the correct support, we know that the corresponding $P(X_1,...,X_k)$ is non-zero. All of the parents of $X_{k+1}$ are elements of $\{X_1,...,X_k\}$; therefore, by definition, the corresponding $P(X_{k+1}|\text{PA}_H(X_{k+1}))$ is also non-zero. Therefore, $P(X_1,...,X_{k+1})=P(X_1,...,X_k)P(X_{k+1}|\text{PA}_H(X_{k+1}))$ is non-zero as required.

	Now, consider some $(x_1, \dotsc, x_k, x_{k+1})$ that \emph{does not} occur in $\mathcal{S}$. There are two possibilities: the set of values  $(x_1, \dotsc, x_k)$ could occur or not occur in $\mathcal{S}$. If $(x_1, \dotsc, x_k)$ also does not occur in $\mathcal{S}$, the proof is simple: by the inductive hypothesis, the corresponding $P(X_1,...,X_k)$ is zero, so $P(X_1,...,X_{k+1})$ has to be zero as required.

	Suppose now that $(x_1, \dotsc, x_k, x_{k+1})$ does not occur in $\mathcal{S}$ but $(x_1, \dotsc, x_k)$ \emph{occurs}. Let $C_1, \dotsc, C_p$ denote the parents of $X_{k+1}$ and $B_1, \dotsc, B_{k-p}$ the remaining variables among $X_1, \dotsc, X_k$, which we know are non-descendants of $X_{k+1}$ in $H$. The DAG $H$ must have the \(d\)-separation relation $X_{k+1} \perp \{B_1, \dotsc, B_{k-p}\}|\{C_1, \dotsc, C_p\}$; since $\mathcal{S}$ is \emph{not} trivially incompatible with $H$, then it must \emph{not} be in conflict with the associated conditional independence relation. Since $(x_1,...,x_k)=(b_1,...,b_{k-p},c_1,...,c_p)$ occurs in $\mathcal{S}$ but $(x_1,...,x_k,x_{k+1})=(b_1,...,b_{k-p},c_1,...,c_p,x_k+1)$ does not occur, we can then conclude that $(x_{k+1}, c_1, \dotsc, c_p)$ \emph{must not occur} in $\mathcal{S}$. Therefore, by definition the associated $P(X_{k+1}|\text{PA}_H(X_{k+1}))=P(X_{k+1}|C_1,...,C_p)$ is zero and hence so is $P(X_1, \dotsc, X_k, X_{k+1})$, as required.
\end{proof}

Accordingly, we have the following upgrade to Theorem~\ref{thm:supports}:

\settheoremtag{Rapid Supports}
\begin{theorem}
	\label{thm:rapid_supports}
	Let $G$ be a DAG. If there is a support $\mathcal{S}$ which is incompatible with $G$ despite \emph{not} being trivially incompatible with $G$, then $G$ is \interesting.
\end{theorem}
\begin{proof}
	We start by noting that if $G$ was \noninteresting, then it would be classically observationally equivalent to some latent-free DAG $H$, per Theorem~\ref{equivalence_evans}. If so, then $H$ and $G$ would (at least!) share the same \(d\)-separation relations, and hence a support would only \emph{not} be trivially incompatible with $G$ if was not trivially incompatible with $H$. But by Lemma~\ref{lem:boring_supports}, if a support is not trivially incompatible with $H$ then it must be compatible with $H$. Since our starting premise is that this support is not compatible with $G$, then by Lemma~\ref{lem:supportsforequiv} it follows that $\mathcal{C}_G \neq \mathcal{C}_H$, and hence $G$ is \interesting.
\end{proof}

Note that Theorem~\ref{thm:rapid_supports} allows us to leverage tools distinct from Fraser's algorithm for assessing support incompatibility. Fraser's algorithm is a \emph{necessary and sufficient} test for support compatibility. For the purposes of Theorem~\ref{thm:rapid_supports}, however, we can instead consider variant algorithms related to Inflation~\cite{Wolfe2019} which cannot certify support compatibility but which can often efficiently detect support incompatibility. Some such algorithms are discussed in Ref~\cite{restivo2022possibilistic}, for example.

It is clear that the application of Theorem~\ref{thm:rapid_supports} is much more efficient than the application of Theorem~\ref{thm:supports}. We can also show that both theorems are equally powerful:

\begin{proposition}
	\label{prop_rapid_supports_subsumes_slow}
	Let $G$ be a \interesting DAG. If the \interestingness of $G$ can be shown via Theorem~\ref{thm:supports}, then it can also be shown via Theorem~\ref{thm:rapid_supports}.
\end{proposition}
\begin{proof}
	First, suppose that there exists a latent-free DAG $H$ such that $\mathcal{I}_H=\mathcal{I}_G$. If $G$ is \interesting, then $C_G\subsetneq C_H$, and hence every support which is incompatible with $G$ must also be incompatible with $H$. If Theorem~\ref{thm:supports} shows the \interestingness of $G$, then there must be a support which is \emph{compatible} with $H$ but not $G$. Since the set of supports incompatible with $H$ is exactly the set of trivially incompatible supports as per Lemmas~\ref{lemma_trivial_support} and~\ref{lem:boring_supports}, it follows that the support which is incompatible with $G$ but not $H$ must \emph{not} be trivially incompatible with $G$. Therefore, the \interestingness of $G$ can be shown by Theorem~\ref{thm:rapid_supports}.

	Now, suppose that there is \emph{no} latent-free DAG that has the same set of \(d\)-separation relations as $G$, i.e. that the set of \(d\)-separation relations of $G$ is \NLF. If $G$ is not maximal, in the proof of Theorem~\ref{thm:elie_esep} we showed a distribution which is not compatible with $G$ (Eq.~\ref{explicit_distribution_esep}). In reality, \emph{all} the distributions which have the same support as this one will be incompatible with $G$. Furthermore, since Eq.~\ref{explicit_distribution_esep} is an element of $\mathcal{I}_G$, this support is not trivially incompatible with $G$. Therefore, the support of the distribution of  Eq.~\ref{explicit_distribution_esep} is an incompatible support which is not trivially incompatible.

	If the set of \(d\)-separation relations of $G$ is \NLF and $G$ \emph{is} maximal, per Theorem~\ref{thm:rapid_dsep_for_interesting} we know that $G$ has one of the eighteen mDAGs of Figure~\ref{fig:d_sep_patterns_compact} as a subgraph. One can explicitly check that these eighteen mDAGs \emph{have} incompatible supports that are not trivially incompatible; therefore, $G$ itself also must have incompatible supports that are not trivially incompatible, namely, by taking all its observed variables outside of the pertinent subgraph to be point distributed. Therefore, this case also falls under the scope of Theorem~\ref{thm:rapid_supports}.
\end{proof}

It is also worth noting that, when there \emph{is} a latent-free DAG $H$ with the same set of \(d\)-separation relations as $G$, Theorem~\ref{thm:rapid_supports} is \emph{constructive}: the distribution $P$ constructed for $H$ in the proof of Lemma~\ref{lem:boring_supports} is such that $P\in \mathcal{I}_G$ yet $P\not\in \mathcal{C}_G$.

An example of a support which is incompatible --- but not \emph{trivially} incompatible --- with the Evans scenario (Figure~\ref{fig_maximal}(b)) is the following:

\begin{equation}
	\label{Evans_support}
	\mathcal{S}_\text{Evans} = \left\{ \begin{aligned} \{C=0,D=0,E=0\}\\ \{C=0,D=1,E=1\} \end{aligned}\right\}
\end{equation}

The Evans scenario does not have any \(d\)-separation relation. Nevertheless, the support $\mathcal{S}_\text{Evans}$ of Equation~\eqref{Evans_support} is not compatible with it. This can be seen by noting that the variable $C$ in $\mathcal{S}$ is associated with a point distribution: it always takes the value $0$. Since in the Evans scenario all the correlation between $D$ and $E$ is established through $C$, it is impossible to have perfect correlation between $D$ and $E$ while $C$ takes a point distribution.

An example of DAG whose \interestingness was first certified in Ref.~\cite{Fraser2020} via the discovery of an incompatible support is presented in Figure~\ref{fig_ex_support}. By means of his eponymous algorithm for compatible supports, Fraser showed that the following support is classically incompatible with the DAG of Figure~\ref{fig_ex_support}\footnote{Note that the support is not \emph{trivially} incompatible, since the DAG in question does not exhibit any \(d\)-separation relations involving observed nodes} :

\begin{equation}
	\label{TC_support}
	\mathcal{S}_\text{\ref{fig_ex_support}} = \left\{\begin{aligned}
		  & \{A=0,B=0,C=0,D=0\} \\
		  & \{A=0,B=0,C=0,D=1\} \\
		  & \{A=0,B=1,C=0,D=0\} \\
		  & \{A=1,B=0,C=1,D=0\}
	\end{aligned}\right\}
\end{equation}

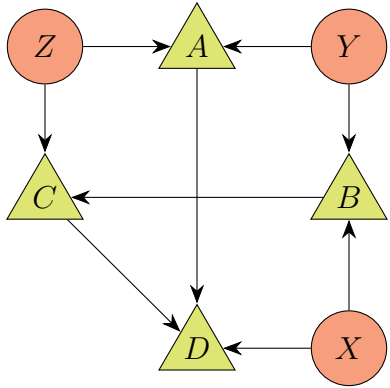
\begin{figure}[htbp]
	\centering
	\begin{tikzpicture}
		\node[q](X) at (4,0){$X$};
		\node[q](Y) at (4,4){$Y$};
		\node[q](Z) at (0,4){$Z$};
		\node[c](A) at (2,4){$A$}
		edge[e](Z)
		edge[e](Y);
		\node[c](B) at (4,2){$B$}
		edge[e](X)
		edge[e](Y);
		\node[c](C) at (0,2){$C$}
		edge[e](Z)
		edge[e](B);
		\node[c](D) at (2,0){$D$}
		edge[e] (X)
		edge[e] (A)
		edge[e] (C);
	\end{tikzpicture}
	\caption{A DAG with 4 observed nodes and 7 total nodes whose \interestingness can be shown by Fraser's algorithm for compatible supports.}
	\label{fig_ex_support}
\end{figure}

When using Theorem~\ref{thm:rapid_supports} to attest \interestingness, we start by checking supports at binary cardinalities of the observed variables. If such an incompatible but not trivially incompatible support is found, we can try to search for such a support at higher cardinalities of the observed variables.

\begin{remark}
	We anticipated that if a DAG has \emph{any} incompatible support (for \emph{any} cardinality), then it seemed likely that we should expect to find \emph{some} incompatible support where all variables have binary cardinality. Indeed, prior to this work, we are not aware of any counterexample. Even the three challenging DAGs identified in Figure 14 of Ref.~\cite{Evans2015} were eventually found to have incompatible supports with merely binary variables. However, this intuition turns out to be misplaced: We identified 4 mDAGs for which  the high-cardinality support of Eq.~\ref{eq_support_card_3222} is identified as incompatible, but where nevertheless \emph{every} support over binary variables is provably compatible. These mDAGs are depicted in Table~\ref{table2new}.
\end{remark}

\begin{equation}
	\label{eq_support_card_3222}
	\mathcal{S}_{\text{for Table~\ref{dags_where_supports_worked}}} = \left\{\begin{aligned}
		  & \{A=0,B=0,C=0,D=0\} \\
		  & \{A=0,B=0,C=1,D=0\} \\
		  & \{A=0,B=1,C=0,D=0\} \\
		  & \{A=1,B=0,C=0,D=0\} \\
		  & \{A=1,B=1,C=0,D=0\} \\
		  & \{A=2,B=0,C=0,D=1\} \\
		  & \{A=2,B=1,C=1,D=0\} \\
	\end{aligned}\right\}
\end{equation}

\begin{table}[h!]
	\centering
	\begin{tabular}{|l|l|}
		\hline
		(a)\hspace{1em} \makecell{\begin{tikzpicture}
				\node[q](X) at (4,2){$X$};
				\node[q](Y) at (2,5){$Y$};
				\node[q](Z) at (0,6){$Z$};
				\node[c](A) at (4,6){$A$}
				edge[e](Z)
				edge[e](Y);
				\node[c](B) at (4,4){$B$}
				edge[e](X)
				edge[e](A);
				\node[c](C) at (2,1){$C$}
				edge[e] (X)
				edge[e] (Y)
				edge[e] (B);
				\node[c](D) at (0,4){$D$}
				edge[e](Z)
				edge[e](B)
				edge[e](C);
			\end{tikzpicture} }   &
		(b)\hspace{1em}  \makecell{\begin{tikzpicture}
				\node[q](X) at (4,2){$X$};
				\node[q](Y) at (2,5){$Y$};
				\node[q](Z) at (0,6){$Z$};
				\node[q](W) at (0,2){$W$};
				\node[c](A) at (4,6){$A$}
				edge[e](Z)
				edge[e](Y);
				\node[c](B) at (4,4){$B$}
				edge[e](X)
				edge[e](A);
				\node[c](C) at (2,1){$C$}
				edge[e](X)
				edge[e](Y)
				edge[e](W)
				edge[e](B);
				\node[c](D) at (0,4){$D$}
				edge[e](Z)
				edge[e](W)
				edge[e](B)
				edge[e](C);
			\end{tikzpicture} }   \\
		\hline
		(c)\hspace{1em}  \makecell{\begin{tikzpicture}
				\node[q](X) at (2,0){$U$};
				\node[q](Y) at (2,5){$Y$};
				\node[q](Z) at (0,6){$Z$};
				\node[c](A) at (4,6){$A$}
				edge[e](Z)
				edge[e](Y);
				\node[c](B) at (4,4){$B$}
				edge[e](X)
				edge[e](A);
				\node[c](C) at (2,2){$C$}
				edge[e](Y)
				edge[e](B);
				\node[c](D) at (0,4){$D$}
				edge[e](Z)
				edge[e](X)
				edge[e](B)
				edge[e](C);
			\end{tikzpicture} } &
		(d)\hspace{1em}   \makecell{ \begin{tikzpicture}
				\node[q](X) at (2,0){$U$};
				\node[q](Y) at (2,5){$Y$};
				\node[q](Z) at (0,6){$Z$};
				\node[q](W) at (0,1){$W$};
				\node[c](A) at (4,6){$A$}
				edge[e](Z)
				edge[e](Y);
				\node[c](B) at (4,4){$B$}
				edge[e](X)
				edge[e](A);
				\node[c](C) at (2,2){$C$}
				edge[e](Y)
				edge[e](W)
				edge[e](B);
				\node[c](D) at (0,4){$D$}
				edge[e](Z)
				edge[e](X)
				edge[e](W)
				edge[e](B)
				edge[e](C);
			\end{tikzpicture} }  \\
		\hline
	\end{tabular}
	\caption{The only 4 mDAGs with 4 observed nodes for which every support over binary observed variables is classically compatible but which are nevertheless provably \interesting by virtue of the higher-cardinality support given in Eq.~\eqref{eq_support_card_3222} being incompatible with all 4 DAGs here.}\label{dags_where_supports_worked}
	\label{table2new}
\end{table}

\subsubsection{Incompatible supports subsumes all other methods}

The methods to show \interestingness that we presented here are not independent of each other. For example, it is clear that Theorem~\ref{new_e_sep} subsumes Theorem~\ref{thm:dsep_for_interesting}, and Proposition~\ref{prop_esep_subsumes} shows that Theorem~\ref{new_e_sep} also subsumes Theorem~\ref{thm:elie_esep}. However, it is more efficient to start by checking maximality of the DAG of interest or the \(d\)-separation relations of latent-free DAGs instead of directly checking all of their \(e\)-separation relations; this is the reason why we presented these results separately.

As it turns out, we can also show that Theorem~\ref{thm:supports} subsumes Theorem~\ref{new_e_sep}: every DAG that can be shown \interesting via its \(e\)-separation relations (through Theorem~\ref{new_e_sep}) can also be shown \interesting via its incompatible supports (through Theorem~\ref{thm:supports}). This result will be presented here as Corollary~\ref{corollary_subsume_interestingness}.  All of the discussion about \(e\)-separation was made here because checking \(e\)-separation relations is in general much faster than checking incompatible supports.

\begin{restatable}[Supports subsumes \(e\)-separation for observational inequivalence]{theorem}{Suppesep}
	\label{th_supp_subsume_inequivalence}
	Let $G$ and $H$ be two DAGs. If their sets of observed \(e\)-separation relations are different, then it is possible to find a support \emph{over binary variables} which is compatible with one of them but incompatible with the other.
\end{restatable}
\begin{proof}
	Presented in Appendix~\ref{appendix_proof_support_esep}.
\end{proof}

Note that Theorem~\ref{th_supp_subsume_inequivalence} is a generic result about observational equivalence between two DAGs, not restricted to the comparison with latent-free DAGs. A direct corollary is:

\begin{corollary}[Supports subsumes \(e\)-separation for \interestingness]
	\label{corollary_subsume_interestingness}
	Let $G$ be a DAG. If there is no latent-free DAG which presents the same set of \(e\)-separation relations as $G$, then there is also no latent-free DAG which presents the same set of classically compatible supports as $G$. In other words, if $G$ can be shown \interesting by Theorem~\ref{new_e_sep}, then it can also be shown \interesting by Theorem~\ref{thm:supports}.
\end{corollary}

Since Proposition~\ref{prop_rapid_supports_subsumes_slow} shows that the rapid supports test (Theorem~\ref{thm:rapid_supports}) can show \interestingness in all the cases covered by Theorem~\ref{thm:supports}, from Corollary~\ref{corollary_subsume_interestingness} we can conclude that Theorem~\ref{thm:rapid_supports} is the most powerful tool we have so far to show \interestingness. It is still an open question whether detection by Theorem~\ref{thm:rapid_supports} is a necessary condition for \interestingness.

Note that Theorem~\ref{thm:rapid_supports} also subsumes Theorem~\ref{thm_setwise_nonmaximal}, as all distributions \emph{with same support} of that in Eq.~\eqref{setwise_correlation} are in $\mathcal{I}$ but not $\mathcal{C}$, and hence every DAG which is provably \interesting via Theorem~\ref{thm:rapid_supports} will also be provably \interesting via Theorem~\ref{thm:rapid_supports}. In particular, merely by considering compatible versus incompatible supports \emph{over two events}, Theorem~\ref{thm:rapid_supports} already subsumes Theorem~\ref{thm_setwise_nonmaximal}.

\section{Computational Results}
\label{sec:results}
We consider the problem of certifying the \interestingness of those mDAGs of 4 observed nodes which are not shown \noninteresting by the \HLP criterion (Corollary~\ref{hlp_criterion}). We start by the method described in Subsection~\ref{sec_nonmaximal}, that says that nonmaximal DAGs are \interesting. Among the 996 mDAGs which are left as \emph{potentially} interesting after applying the \HLP criterion, we find that 810 are nonmaximal. Therefore, nonmaximality seems to be a powerful tool to show \interestingness; at this stage, we are left with only 186 unresolved mDAGs after this preliminary filtering.

We then exploit the \(d\)-separation test for \interestingness per Subsection~\ref{subsec:dsep}. That is, among those 186 remaining mDAGs we filter out any mDAGs which possess observed \(d\)-separation relations \emph{not} matching those of some latent-free DAG (\NLF \(d\)-separation relations). We find that 168 mDAGs remain as-yet unresolved --- the 18 mDAGs that are shown \interesting at this stage are the ones presented in Figure~\ref{fig:d_sep_patterns_compact}.

The \(e\)-separation test for \interestingness presented in Subsection~\ref{subsec:esep} does not resolve \emph{any} of these 168 unresolved cases. As mentioned, it is still an open problem whether the nonmaximality test and the \(d\)-separation test together will always subsume the \(e\)-separation test. We then turn to the method of supports analysis per Subsection~\ref{subsec:supports}, which ultimately leaves us with only 3 remaining unresolved mDAGs.

More specifically, by considering supports with \emph{binary} cardinalities of the observed variables we were able to certify the \interestingness of $161$ out of the $168$ remaining mDAGs. We could not, however, find \emph{any} classically incompatible supports for the remaining 7 mDAGs when only considering binary cardinality variables. But by increasing the cardinality of variable $A$ to be three we were able to identify a support that is incompatible --- but not \emph{trivially} incompatible --- with the four mDAGs in Table~\ref{table2new}, hence certifying their \interestingness.  Said support is explicitly reproduced in Eq.~\eqref{eq_support_card_3222}; it is obviously not trivially incompatible with any of the mDAGs in Table~\ref{table2new}, since none of those four mDAGs exhibits any \(d\)-separation relations over its observed nodes. The remaining 3 mDAGs which we were \emph{unable} to resolve as \interesting via supports considerations --- up to computational tractability limits --- are depicted in Table~\ref{table3}.

The exact number of mDAGs that we are able to characterize as \interesting or \noninteresting at each stage is summarized in Table~\ref{tab:introsummary}.

\begin{table}
	\centering
	\begin{tabular}{|l|}
		\hline
		(a)\hspace{1em}  \makecell{\begin{tikzpicture}[yscale=0.85, xscale=1.184]
				\node[q](X) at (2,-1){$X$};
				\node[q](Y) at (2,4.5){$Y$};
				\node[q](Z) at (0,6){$Z$};
				\node[c](A) at (4,6){$A$}
				edge[e](Z)
				edge[e](Y);
				\node[c](B) at (4,3){$B$}
				edge[e](X)
				edge[e](A);
				\node[c](C) at (2,1){$C$}
				edge[e] (X)
				edge[e] (Y)
				edge[e] (B);
				\node[c](D) at (0,3){$D$}
				edge[e](Z)
				edge[e] (X)
				edge[e](B)
				edge[e](C);
			\end{tikzpicture}} \\
		\hline
		(b)\hspace{1em}  \makecell{\begin{tikzpicture}[yscale=0.85, xscale=1.184]
				\node[q](X) at (4,0){$W$};
				\node[q](Y) at (2,4.5){$Y$};
				\node[q](Z) at (0,6){$Z$};
				\node[q](W) at (2,-1){$X$};
				\node[c](A) at (4,6){$A$}
				edge[e](Z)
				edge[e](Y);
				\node[c](B) at (4,3){$B$}
				edge[e](X)
				edge[e](W)
				edge[e](A);
				\node[c](C) at (2,1){$C$}
				edge[e] (X)
				edge[e] (Y)
				edge[e] (B);
				\node[c](D) at (0,3){$D$}
				edge[e](Z)
				edge[e](B)
				edge[e](W)
				edge[e](C);
			\end{tikzpicture}} \\
		\hline
		(c)\hspace{1em}  \makecell{\begin{tikzpicture}[yscale=0.85, xscale=1.184]
				\node[q](X) at (4,0){$W$};
				\node[q](Y) at (2,4.5){$Y$};
				\node[q](Z) at (0,6){$Z$};
				\node[q](W) at (2,-1){$X$};
				\node[q](V) at (0,0){$U$};
				\node[c](A) at (4,6){$A$}
				edge[e](Z)
				edge[e](Y);
				\node[c](B) at (4,3){$B$}
				edge[e](X)
				edge[e](W)
				edge[e](A);
				\node[c](C) at (2,1){$C$}
				edge[e](X)
				edge[e](Y)
				edge[e](V)
				edge[e] (B);
				\node[c](D) at (0,3){$D$}
				edge[e](Z)
				edge[e](B)
				edge[e](W)
				edge[e](V)
				edge[e](C);
			\end{tikzpicture}} \\
		\hline
	\end{tabular}
	\caption{The mDAGs of 4 observed nodes whose \interestingness remains unresolved.}
	\label{table3}
\end{table}

\subsection{On the potential \interestingness of the remaining 3 mDAGs}

The 3 as-yet unresolved mDAGs with 4 observed nodes are depicted in Table~\ref{table3}. For these 3 mDAGs we could not find \emph{any} incompatible supports, at least up to the small cardinalities of the observed nodes that we checked. Searching for incompatible supports at higher cardinalities of the observed nodes using Fraser's algorithm is computationally expensive, as the algorithm's complexity increases significantly on increasing the cardinalities. Perhaps future acceleration of Fraser's algorithm may allow us to probe supports for higher cardinalities. For the present work, however, we considered one final attempt to prove\footnote{The entropic technique is predicated on trying to prove that the Shannon-type entropic inequalities constraining $\mathcal{C}_G$ define a hypercone strictly in the interior of the hypercone given by the Shannon-type entropic inequalities constraining $\mathcal{I}_G$. Even when such a finding can be established, however, it need not mean that $\mathcal{C}_G\subsetneq \mathcal{I}_G$, though that is \emph{almost certainly} the case. The loophole --- however implausible --- is that perhaps the \emph{true} entropy cone constraining $\mathcal{I}_G$ actually coincides with (or is interior to) the projection of the Shannon cone associated with $\mathcal{C}_G$. See Appendix E of Ref.~\cite{HLP_2014}.} the \interestingness of these 3 mDAGs, namely, by exploring entropic inequalities. For a more comprehensive introduction to entropic inequalities and Shannon cones see Refs.~\cite{Weilenmann2017, Miklin2017}.

In particular, we attempted to isolate some Shannon-type inequalities that constrain $\mathcal{C}_G$, but are not Shannon-type inequalities for $\mathcal{I}_G$. This analysis is also computationally expensive --- often intractable --- as generating the Shannon-type inequalities corresponding to $\mathcal{C}$ is accomplished via linear quantifier elimination. The complexity of the most common algorithm for performing linear quantum elimination is doubly exponential in the number of eliminated variables, though alternative algorithms have different complexities~\cite{PolytopeProjectionReview}.

We can nevertheless certify that the Shannon cone corresponding to $\mathcal{C}_G$ and $\mathcal{I}_G$ are indistinguishable for these 3 remaining mDAGs \emph{without} explicitly constructing the Shannon cone for $\mathcal{C}_G$. We do so as follows:
\begin{enumerate}
	\item From all the Shannon type inequalities corresponding to $\mathcal{I}_G$, generate the extremal rays of this cone.
	\item Check whether each of those extremal rays is \emph{implicitly} contained in the Shannon cone corresponding to $\mathcal{C}_G$ by asking if the Shannon-type inequalities over \emph{all} the variables (thus not using linear quantifier elimination) corresponding to $\mathcal{C}$ are \emph{satisfiable} by the given extremal ray. If every extremal ray of the Shannon cone of $\mathcal{I}_G$ is contained in the Shannon cone corresponding to $\mathcal{C}_G$, then those two Shannon cones coincide.
\end{enumerate}

For the 3 mDAGs in Table~\ref{table3} we find that the Shannon cones corresponding to $\mathcal{C}_G$ and $\mathcal{I}_G$ are the same. That is, we were unable to find any valid Shannon type inequality for $\mathcal{C}_G$ that is not also a Shannon type inequality for $\mathcal{I}_G$ for these 3 mDAGs. Thus, entropic methods are incapable of proving the \interestingness of these 3 mDAGs, unless perhaps we explore non-Shannon-type inequalities or entropic inequalities involving non-Shannon entropies.

\section{Conclusions}
\label{con}
In this work, we contributed to causal investigation by categorizing which causal structures of 4 observed nodes present inequality constraints or not. To do so, we developed a plethora of techniques to prove that a causal structure is \interesting (has inequality constraints), while we used one single technique to prove that a causal structure is \noninteresting: the \HLP criterion.

As can be seen from Table~\ref{tab:introsummary}, out of the 2809 mDAGs with 4 observed nodes, the \HLP criterion shows that that 1813 are \noninteresting. Out of the remaining 996 mDAGs, our techniques showed 993 of them to be \interesting, while we are still uncertain about the status of 3 mDAGs (presented in Figure~\ref{table3}). While these 3 remaining mDAGs are still potential counter-examples to the \HLP conjecture (which says that the \HLP criterion is necessary and sufficient for \noninterestingness), we believe that our numerical results are a hint towards the validity of this conjecture. A truly thorough analysis of all mDAGs with 5 observed nodes proved to be quite computationally demanding. Nevertheless, in Appendix~\ref{app:5nodes} we show that --- among those 5-node mDAGs which the \HLP criterion fails to certify as \interesting --- \emph{at least} 99\% are \interesting, which we again elect to interpret as at least \emph{consistent} with the \HLP conjecture.

It is also interesting to note that all of our techniques to show \interestingness give explicit constructions of distributions which are in $\mathcal{I}_G$ but not in $\mathcal{C}_G$, i.e., respect the conditional independence constraints of DAG $G$ but not its inequality constraints. Theorem~\ref{thm:elie_esep} is related to the construction of Eq.~\ref{explicit_distribution_esep}, as well as Theorem~\ref{thm_setwise_nonmaximal} is related to the construction of Eq.~\ref{setwise_correlation}. Theorem~\ref{thm:rapid_supports} is constructive whenever there \emph{is} a latent-free DAG $H$ with the same set of \(d\)-separation relations as $G$ (such construction can then be found in the proof of Lemma~\ref{lem:boring_supports}). If $G$ is maximal but there is \emph{no} latent-free DAG $H$ with the same set of \(d\)-separation relations as $G$, then Theorem~\ref{thm:rapid_dsep_for_interesting} says that $G$ has one of the eighteen DAGs of Figure~\ref{fig:d_sep_patterns_compact} as a subgraph. In the end of Appendix~\ref{app_rapid_dsep}, an explicit distribution which is in $\mathcal{I}_G$ but not in $\mathcal{C}_G$ for these eighteen DAGs is presented: it is the uniform distribution over the events in the Popescu-Rohrlich support presented in Eq.~\eqref{eq_support_Bell}.

In particular, our Theorem~\ref{thm:elie_esep} which showed itself to be very powerful in proving \interestingness, is a corrected version of the \(e\)-separation theorem of~\cite{Pienaar2017} (as discussed in Appendix~\ref{app:pienaar}).

By showing practical tools to attest that a causal structure presents inequality constraints, this work simultaneously contributes to purely classical causal inference and advances the question of which causal scenarios might exhibit quantum or post-quantum advantage.

\section*{Acknowledgements}
S.K would like to thank Roger Colbeck for insightful discussions. Research at Perimeter Institute is supported in part by the Government of Canada through the Department of Innovation, Science and Economic Development and by the Province of Ontario through the Ministry of Colleges and Universities.

\newpage
\addcontentsline{toc}{section}{References}
\bibliographystyle{apsrev4-2-wolfe}
\setlength{\bibsep}{3pt plus 3pt minus 2pt}
\nocite{apsrev42Control}
\bibliography{refs}

\begin{appendices}
	\section{Pienaar's \(e\)-separation theorem is incorrect}\label{app:pienaar}

	In Ref.~\cite[Lemma 2]{Pienaar2017} Pienaar presented the following claim, which we demonstrate to be incorrect:
	\begin{theorem}[\textbf{Pienaar's incorrect \(e\)-separation theorem}]
		\label{th3}
		Let $G$ be a DAG, and let $\boldsymbol{X}$, $\boldsymbol{Y}$, $\boldsymbol{Z}$ and $\boldsymbol{W}$ be disjoint sets of nodes of $G$ such that none of the nodes of $\boldsymbol{Z}$ is a descendant of a node in $\boldsymbol{W}$ and $(\boldsymbol{X}\eperp \boldsymbol{Y}|\boldsymbol{Z})_{del_{\boldsymbol{W}}}$. Then $G$ is \interesting \emph{if and only if} the \(d\)-separation relations of $G$ \emph{do not} include any relations of the form $(\boldsymbol{X}\perp \boldsymbol{Y} |\boldsymbol{Z}\boldsymbol{S})$, where $\boldsymbol{S}$ is a subset of $\boldsymbol{W}$ .
	\end{theorem}

	If this theorem were true, then the problem of classifying \interestingness would be completely solved whenever a suitable \(e\)-separation is present, as Theorem~\ref{th3} claims to provides a necessary and sufficient condition for such DAGs that is simple to check. However, the condition turns out to be \emph{neither} necessary \emph{nor} sufficient.

	Firstly, note that there are plenty of well-known graphs which are \interesting despite not satisfying the conditions of Theorem~\ref{th3}. Examples include the Bell DAG of Figure~\ref{fig_Bell_DAG}, as well as the mDAG in Figure~\ref{fig_maximal}a, among many others.

	Pienaar notably included similar examples in his own work~\cite{Pienaar2017}, and therefore we believe the inclusion of the \enquote{only if} language in Theorem~\ref{th3} was an oversight, in that Pienaar himself never actually intended to communicate that, but rather only that the \(e\)-separation relation ${(\boldsymbol{X}\eperp \boldsymbol{Y}|\boldsymbol{Z})_{del_{\boldsymbol{W}}}}$ would automatically follow from the \(d\)-separation relation ${(\boldsymbol{X}\perp \boldsymbol{Y} |\boldsymbol{Z}\boldsymbol{S})}$ if that \(d\)-separation relation was present in the graph and where $\boldsymbol{S}$ is a subset of $\boldsymbol{W}$.

	Regardless, the \enquote{if} direction in Theorem~\ref{th3} \emph{also} turns out to be invalid, though it is a bit more subtle. Consider the DAG in Figure~\ref{figaddnew}a; Pienaar~\cite{Pienaar2017} uses it as an example of a DAG that is deemed (in this case, \emph{correctly} deemed) as \interesting pursuant to Theorem~\ref{th3}. In this DAG, $(F \eperp D | C)_{del_E}$ holds, and node $C$ is not a descendant of node $E$, and neither $(F \perp D | C)$ nor $(F \perp D |CE)$ holds true. Therefore, Theorem~\ref{th3}  classifies the DAG of Figure~\ref{figaddnew}a as \interesting. But now consider the DAG of Figure~\ref{figaddnew}b which has the same \emph{structure} as the DAG in Figure~\ref{figaddnew}a, with the only difference being that the DAG in Figure~\ref{figaddnew}b has no latent variables. All the conditions of Theorem~\ref{th3} are again met, however, when applied to Figure~\ref{figaddnew}.  Again, $(F \eperp D | C)_{del_E}$ holds, and node $C$ is not a descendant of node $E$, and neither $(F \perp D | C)$ nor $(F \perp D |CE)$  hold true. So this DAG is again --- but this time, wrongly --- characterised as \interesting by Pienaar's theorem. As previously discussed, all latent-free DAGs are \noninteresting!

	In fact, \emph{every} DAG for which the conditions of Theorem~\ref{th3} hold can be converted into a different latent-free DAG for which the conditions of the theorem would continue to hold by making all the nodes observed. That is, for every DAG which is correctly classified as \interesting by the (invalid) condition formulated as Theorem~\ref{th3} one can construct a latent-free counterexample to Theorem~\ref{th3}.

	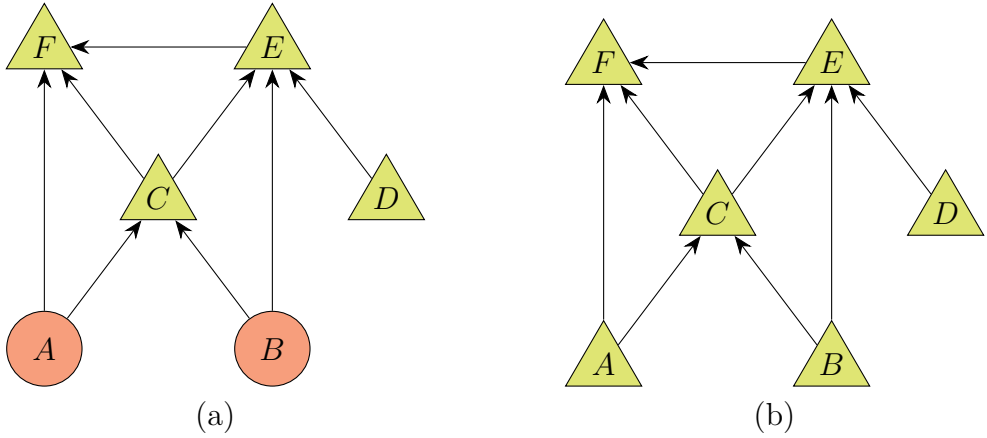
\begin{figure}[h!]
		\centering
		\begin{tabular}{ccc}
			\begin{tikzpicture}
				\node[q](A) at (0,0){$A$};
				\node[q](B) at (3,0){$B$};
				\node[c](D) at (4.5,2){$D$};
				\node[c](C) at (1.5,2){$C$}
				edge[e] (A)
				edge[e] (B);
				\node[c](E) at (3,4){$E$}
				edge[e] (B)
				edge[e] (C)
				edge[e] (D);
				\node[c](D) at (0,4){$F$}
				edge[e] (A)
				edge[e] (C)
				edge[e] (E);
			\end{tikzpicture} & \hspace{1cm} &

			\begin{tikzpicture}
				\node[c](A) at (0,0){$A$};
				\node[c](B) at (3,0){$B$};
				\node[c](D) at (4.5,2){$D$};
				\node[c](C) at (1.5,2){$C$}
				edge[e] (A)
				edge[e] (B);
				\node[c](E) at (3,4){$E$}
				edge[e] (B)
				edge[e] (C)
				edge[e] (D);
				\node[c](D) at (0,4){$F$}
				edge[e] (A)
				edge[e] (C)
				edge[e] (E);
			\end{tikzpicture} \\
			(a)                         &              & (b)
		\end{tabular}
		\caption{(a) depicts a scenario characterised correctly by Pienaar's theorem as \interesting, whereas (b) depicts an \noninteresting scenario which Pienaar's theorem incorrectly deems \interesting.}
		\label{figaddnew}
	\end{figure}

	The flaw with Pienaar's proof of the \enquote{if} direction is that it invokes a probability distribution that may not actually be in $\mathcal{I}$, in that the constructed distribution may be inconsistent with some \(d\)-separation relation \emph{not mentioned in the statement of Theorem~\ref{th3}}. For example, in Figure~\ref{figaddnew}, Pienaar's proof would invoke a distribution which posits perfect correlation between $D$ and $F$ while all other variables are point-distributed. While such a distribution is consistent with \emph{all} the observable \(d\)-separation relations of Figure~\ref{figaddnew}a, it is nevertheless \emph{inconsistent} with the \(d\)-separation relation $(F \perp D|BCE)$ exhibited by Figure~\ref{figaddnew}b.

	This loophole in Pienaar's proof can be closed by \emph{strengthening the conditions of Pienaar's theorem}, namely, to exclude $(F \perp D | \boldsymbol{S})$ for \emph{any} subset $\boldsymbol{S}$ of the the remaining observed nodes. In other words,
	\begin{theorem}[\textbf{A corrected version of Pienaar's \(e\)-separation theorem}]
		\label{th3_corrected}
		Let $G$ be a DAG, and let $\boldsymbol{X}$, $\boldsymbol{Y}$, $\boldsymbol{Z}$ and $\boldsymbol{W}$ be disjoint sets of nodes of $G$ such that none of the nodes of $\boldsymbol{Z}$ is a descendant of a node in $\boldsymbol{W}$ and $(\boldsymbol{X}\eperp \boldsymbol{Y}|\boldsymbol{Z})_{del_{\boldsymbol{W}}}$. Then $G$ is \interesting if the \(d\)-separation relations of $G$ \emph{do not} include any relations of the form $(\boldsymbol{X}\perp \boldsymbol{Y} |\boldsymbol{S})$, where $\boldsymbol{S}$ may be any subset of the observed nodes of $G$ other than $\boldsymbol{X}\cup\boldsymbol{Y}$.
	\end{theorem}
	However, a graph can only exclude relations of the form $\boldsymbol{X}\perp \boldsymbol{Y} |\boldsymbol{S}$ for any $\boldsymbol{S}$ if for every pair of singleton nodes $\{X,Y\}$ such that $X\in\boldsymbol{X}$ and $Y\in \boldsymbol{Y}$ it holds that $(X\perp Y|\boldsymbol{S})$. On the other hand, a pair of nodes $\{X,Y\}$ can only be \(e\)-separated by $\boldsymbol{Z}$ upon the removal of $\boldsymbol{W}$ if $X$ and $Y$ are not \emph{adjacent}. Consequently, a DAG can \emph{only} be certified as \interesting by Theorem~\ref{th3_corrected} if there exists a pair of \emph{\(d\)-unseparable} but nevertheless non-adjacent nodes. Which is to say, Theorem~\ref{th3_corrected} is ultimately \emph{equivalent}\footnote{Plainly Theorem~\ref{thm:elie_esep} can be seen as a special case of Theorem~\ref{th3_corrected}, by taking $\boldsymbol{Z}$ to be empty, choosing  $\boldsymbol{X}$ and $\boldsymbol{Y}$ to be singleton-sized sets, and letting $\boldsymbol{W}$ be the set of all observed nodes other than $\boldsymbol{X}$ and $\boldsymbol{Y}$.} to Theorem~\ref{thm:elie_esep} of the main text, albeit obfuscated.

	\section{Revisiting the Skeleton Condition}\label{app:skeleton}

	In the main test we noted (Theorem~\ref{thm:elie_esep}) that every nonmaximal DAG is \interesting. Here we revisit that finding, and contrast is with a prior result due to Pienaar~\cite{Pienaar2017}.

	Firstly, we recall that the set of adjacent node pairs in a DAG is conventionally referred to as the DAG's \emph{skeleton}.
	\begin{definition}[\textbf{Skeleton of a DAG.}]
		The skeleton of a DAG $G$ is the undirected graph wherein $A$ and $B$ are adjacent in $G$'s skeleton --- denoted $A\undirectededge B$  --- if and only if $A$ and $B$ are adjacent in $G$.
	\end{definition}

	Notably, the skeleton of a DAG can be used to witness observational inequivalence. In particular, Evans has shown that if a pair of nodes is adjacent in some mDAG $G$ but not in some other mDAG $H$, then $\mathcal{C}_G\neq\mathcal{C}_H$~\cite{Evans2015, Evans2012}. We formally express this idea here as:
	\begin{lemma}[Skeleton condition for observational equivalence and dominance]
		\label{lem:skeleton_equivalence}
		For any pair of DAGs $G$ and $H$, $\mathcal{C}_G=\mathcal{C}_H$ only if the skeletons of $G$ and $H$ are the same, i.e., they agree on all observable node (non)adjacencies. Additionally, $C_G\subset C_H$ only if every pair of adjacent nodes in $G$ is also adjacent in $H$.
	\end{lemma}
	As with any necessary conditional for observational equivalence, Lemma~\ref{lem:skeleton_equivalence} can then be utilized to formulate a sufficient condition for \interestingness. Pienaar did exactly that when he formulated the following Theorem in Ref.~\cite[Theorem 1]{Pienaar2017}:
	\begin{theorem}[Pienaar's skeleton condition for \interestingness]\label{thm:original_skel}
		Let $G$ be the DAG we need to check for \interestingness. Suppose that one can find some other DAG $H$ with the same observed conditional independences as $G$ and which is certainly known to be \noninteresting, i.e. $\mathcal{I}_G=\mathcal{I}_H$ and $\mathcal{I}_H=\mathcal{C}_H$. Then, if the skeletons of $G$ and $H$ are different, evidently $\mathcal{C}_G\neq \mathcal{C}_H$ per Lemma \ref{lem:skeleton_equivalence}, and hence $\mathcal{C}_G\neq \mathcal{I}_G$, i.e. DAG $G$ is \interesting.
	\end{theorem}

	Theorem~\ref{thm:original_skel} is both correct and useful. Note, however, that it is \emph{superseded} by the Theorem~\ref{thm:elie_esep} presented in the main text. Suppose that there \emph{does} exists some latent-free DAG $H$ that agrees with the observed \(d\)-separation relations of $G$. Then, among other things, $G$ and $H$ evidently agree on \(d\)-unseparable observed node pairs. Since the skeleton of $H$ is defined by $H$'s \(d\)-unseparable node pairs~\cite[Prop. 3.19]{Richardson2002}, the only way for the skeletons of $G$ and $H$ to differ is if $G$ contains some pair of \emph{nonadjacent} \(d\)-unseparable observed nodes. Consequently, Theorem~\ref{thm:original_skel} can only certify a DAG as \interesting when the DAG's \interestingness is also certifiable by Theorem~\ref{thm:elie_esep}.

	Thus, \emph{both} of Pienaar's conditions for \interestingness are subsumed by Theorem~\ref{thm:elie_esep} here, at least after adjusting Pienaar's condition based on \(e\)-separation as per Appendix~\ref{app:pienaar}.

	We know that there are DAGs for which one cannot find any DAG $H$ with the required properties to make use of Theorem~\ref{thm:original_skel}; see Appendix~\ref{app_rapid_dsep} for examples. In light of Theorem~\ref{equivalence_evans} we would want to reformulate Theorem~\ref{thm:original_skel} in a manner that is clearly (strictly) more powerful, which removes the caveat about finding such an $H$.
	\begin{theorem}[Improved skeleton condition for \interestingness]\label{thm:improved_skel}
		Let $G$ be the DAG we need to check for \interestingness. Consider all latent-free graphs which share the same skeleton as that of $G$. If no latent-free graph within that set furthermore matches the observed \(d\)-separations of $G$, then $G$ is \interesting.
	\end{theorem}
	\begin{proof}Theorem~\ref{thm:improved_skel} follows directly from combining Lemma~\ref{lem:skeleton_equivalence} and Eq.~\eqref{eq_different_I} with Theorem~\ref{equivalence_evans}.
	\end{proof}

	As it turns out, however, Theorem~\ref{thm:improved_skel} is \emph{equivalent} to the conjunction of Theorems~\ref{thm:elie_esep} and ~\ref{thm:dsep_for_interesting}. First, we can see that  Theorem~\ref{thm:elie_esep} is a special case of Theorem~\ref{thm:improved_skel}. If $G$ is nonmaximal, then it has at least one non-adjacent pair of observed nodes which is \(d\)-unseparable. This implies that all the latent-free DAGs which have the same skeleton (same adjacency structure) as $G$ have a different set of \(d\)-separable pairs of nodes than that of $G$, because all latent-free DAGs are maximal~\cite[Prop. 3.19]{Richardson2002}. This then implies that they have a different set of \(d\)-separation relations than those of $G$, so Theorem~\ref{thm:improved_skel} witnesses all nonmaximal DAGs as \interesting. Secondly, it is easy to see that  Theorem~\ref{thm:dsep_for_interesting} is also a special case of Theorem~\ref{thm:improved_skel}: if there are no latent-free DAGs that match the set of \(d\)-separation relations of $G$, then in particular there are no latent-free DAGs that match the the set of \(d\)-separation relations \emph{and} the skeleton of $G$.

	\sloppy Next, we prove the inverse, that Theorems~\ref{thm:elie_esep} and~\ref{thm:dsep_for_interesting} \emph{together} subsume Theorem~\ref{thm:improved_skel}. Let $G$ be a DAG which is shown \interesting by Theorem~\ref{thm:improved_skel}. If $G$ is nonmaximal, then it is also shown \interesting by Theorem~\ref{thm:elie_esep}. What remains, then, is to show that there are no \emph{maximal} DAGs which \emph{cannot} be proven \interesting via Theorem~\ref{thm:dsep_for_interesting} but \emph{can} be seen as \interesting via Theorem~\ref{thm:improved_skel}. But if the DAG $G$ is both maximal \emph{and} shares the same \(d\)-separation relations as some latent-free DAG $H$, it \emph{automatically follows} that $G$ and $H$ agree on their skeletons as well. After all, $G$ and $G$ are \emph{both} maximal, which means that their skeletons are dictated by their respective \(d\)-unseparable node pairs, which are identical.

	Note that Theorem~\ref{thm:improved_skel} is also subsumed by Theorem~\ref{new_e_sep}, since agreeing on \(e\)-separation relations implies agreeing on adjacencies (i.e., skeleton) as well as on \(d\)-separation relations.

	\section{Rapid \(d\)-separation test}
	\label{app_rapid_dsep}

	In Section~\ref{subsec:dsep}, we described a method to show the \interestingness of a DAG when its set of observed \(d\)-separation relations is \NLF, i.e. it does not match those of any latent-free DAG; this is encoded in Theorem~\ref{thm:dsep_for_interesting}. To apply this method in practice, we indicated that one can construct all of the latent-free DAGs that have the same number of observed nodes as the DAG in question, and then compare their sets of \(d\)-separation relations. However, when the DAG is \emph{maximal} (as per Definition~\ref{def_maximal}) there is a faster way to apply this method, which will be described now. When the DAG is \emph{not} maximal, its \interestingness is automatically attested by Theorem~\ref{thm:elie_esep}.

	When the DAG is maximal, to see whether its set of \(d\)-separation relations is \NLF one just needs to check whether it has one of the DAGs of Figure~\ref{fig:d_sep_patterns_compact} as a subgraph. This was recognized in Evans' own proof of Theorem~\ref{equivalence_evans}~\cite{Evans_2022}, as we argue in the proof of the following theorem:

	\begin{flushleft}
		\begin{figure}[h!]
			\centering
			\begin{tabular}{cccc}
				\begin{tikzpicture}[scale=0.55]
					\node[q1](L) at (2,0){\small $\Lambda$};
					\node[c1](S) at (0,-1.5){\small $X$};
					\node[c1](T) at (4,-1.5){\small $Y$};
					\node[c1](X) at (0,1.5){\small $A$}
					edge[e](L)
					edge[e](S);
					\node[c1](Y) at (4,1.5){\small $B$}
					edge[e](L)
					edge[e](T);
				\end{tikzpicture}
				&
				\begin{tikzpicture}[scale=0.55]
					\node[q1](O) at (-1,0){\small $\Omega$};
					\node[q1](L) at (2,0){\small $\Lambda$};
					\node[c1](S) at (0,-1.5){\small $X$}
					edge[e](O);
					\node[c1](T) at (4,-1.5){\small $Y$};
					\node[c1](X) at (0,1.5){\small $A$}
					edge[e](O)
					edge[e](L)
					edge[e](S);
					\node[c1](Y) at (4,1.5){\small $B$}
					edge[e](L)
					edge[e](T);
				\end{tikzpicture}
				&
				\begin{tikzpicture}[scale=0.55]
					\node[q1](O) at (-1,0){\small $\Omega$};
					\node[q1](L) at (2,0){\small $\Lambda$};
					\node[c1](S) at (0,-1.5){\small $X$}
					edge[e](O);
					\node[c1](T) at (4,-1.5){\small $Y$};
					\node[c1](X) at (0,1.5){\small $A$}
					edge[e](O)
					edge[e](L);
					\node[c1](Y) at (4,1.5){\small $B$}
					edge[e](L)
					edge[e](T);
				\end{tikzpicture}
				\\
				(a) & (b) & (c) \\
				\begin{tikzpicture}[scale=0.55]
					\node[q1](O) at (-1,0){\small $\Omega$};
					\node[q1](Z) at (5,0){\small $\Sigma$};
					\node[q1](L) at (2,0){\small $\Lambda$};
					\node[c1](S) at (0,-1.5){\small $X$}
					edge[e](O);
					\node[c1](T) at (4,-1.5){\small $Y$}
					edge[e](Z);
					\node[c1](X) at (0,1.5){\small $A$}
					edge[e](O)
					edge[e](L)
					edge[e](S);
					\node[c1](Y) at (4,1.5){\small $B$}
					edge[e](T)
					edge[e](Z)
					edge[e](L);
				\end{tikzpicture}
				&
				\begin{tikzpicture}[scale=0.55]
					\node[q1](O) at (-1,0){\small $\Omega$};
					\node[q1](Z) at (5,0){\small $\Sigma$};
					\node[q1](L) at (2,0){\small $\Lambda$};
					\node[c1](S) at (0,-1.5){\small $X$}
					edge[e](O);
					\node[c1](T) at (4,-1.5){\small $Y$}
					edge[e](Z);
					\node[c1](X) at (0,1.5){\small $A$}
					edge[e](O)
					edge[e](L);
					\node[c1](Y) at (4,1.5){\small $B$}
					edge[e](T)
					edge[e](Z)
					edge[e](L);
				\end{tikzpicture}
				&
				\begin{tikzpicture}[scale=0.55]
					\node[q1](O) at (-1,0){\small $\Omega$};
					\node[q1](Z) at (5,0){\small $\Sigma$};
					\node[q1](L) at (2,0){\small $\Lambda$};
					\node[c1](S) at (0,-1.5){\small $X$}
					edge[e](O);
					\node[c1](T) at (4,-1.5){\small $Y$}
					edge[e](Z);
					\node[c1](X) at (0,1.5){\small $A$}
					edge[e](O)
					edge[e](L);
					\node[c1](Y) at (4,1.5){\small $B$}
					edge[e](Z)
					edge[e](L);
				\end{tikzpicture}
				\\ (d) & (e) & (f) \\

				\begin{tikzpicture}[scale=0.55]                                        \node[q1](L) at (2,0){\small $\Lambda$};
					\node[c1](S) at (0,-1.5){\small $X$};
					\node[c1](T) at (4,-1.5){\small $Y$}
					edge[e](S);
					\node[c1](X) at (0,1.5){\small $A$}
					edge[e](L)
					edge[e](S);
					\node[c1](Y) at (4,1.5){\small $B$}
					edge[e](L)
					edge[e](T);
				\end{tikzpicture}
				&
				\begin{tikzpicture}[scale=0.55]
					\node[q1](O) at (-1,0){\small $\Omega$};
					\node[q1](L) at (2,0){\small $\Lambda$};
					\node[c1](S) at (0,-1.5){\small $X$}
					edge[e](O);
					\node[c1](T) at (4,-1.5){\small $Y$}
					edge[e](S);
					\node[c1](X) at (0,1.5){\small $A$}
					edge[e](L)
					edge[e](O)
					edge[e](S);
					\node[c1](Y) at (4,1.5){\small $B$}
					edge[e](L)
					edge[e](T);
				\end{tikzpicture}
				&
				\begin{tikzpicture}[scale=0.55]
					\node[q1](O) at (-1,0){\small $\Omega$};
					\node[q1](L) at (2,0){\small $\Lambda$};
					\node[c1](S) at (0,-1.5){\small $X$}
					edge[e](O);
					\node[c1](T) at (4,-1.5){\small $Y$}
					edge[e](S);
					\node[c1](X) at (0,1.5){\small $A$}
					edge[e](L)
					edge[e](O);
					\node[c1](Y) at (4,1.5){\small $B$}
					edge[e](L)
					edge[e](T);
				\end{tikzpicture}
				\\ (g) & (h) & (i) \\
				\begin{tikzpicture}[scale=0.55]
					\node[q1](G) at (2,-3){\small $\Gamma$};
					\node[q1](L) at (2,0){\small $\Lambda$};
					\node[c1](S) at (0,-1.5){\small $X$}
					edge[e](G);
					\node[c1](T) at (4,-1.5){\small $Y$}
					edge[e](G);
					\node[c1](X) at (0,1.5){\small $A$}
					edge[e](L)
					edge[e](S);
					\node[c1](Y) at (4,1.5){\small $B$}
					edge[e](L)
					edge[e](T);
				\end{tikzpicture}
				&
				\begin{tikzpicture}[scale=0.55]
					\node[q1](G) at (2,-3){\small $\Gamma$};
					\node[q1](L) at (2,0){\small $\Lambda$};
					\node[c1](S) at (0,-1.5){\small $X$}
					edge[e](G);
					\node[c1](T) at (4,-1.5){\small $Y$}
					edge[e](G)
					edge[e](S);
					\node[c1](X) at (0,1.5){\small $A$}
					edge[e](L)
					edge[e](S);
					\node[c1](Y) at (4,1.5){\small $B$}
					edge[e](L)
					edge[e](T);
				\end{tikzpicture}
				&

				\begin{tikzpicture}[scale=0.55]
					\node[q1](G) at (2,-3){\small $\Gamma$};
					\node[q1](O) at (-1,0){\small $\Omega$};
					\node[q1](Z) at (5,0){\small $\Sigma$};
					\node[q1](L) at (2,0){\small $\Lambda$};
					\node[c1](S) at (0,-1.5){\small $X$}
					edge[e](O)
					edge[e](G);
					\node[c1](T) at (4,-1.5){\small $Y$}
					edge[e](Z)
					edge[e](G);
					\node[c1](X) at (0,1.5){\small $A$}
					edge[e](O)
					edge[e](L);
					\node[c1](Y) at (4,1.5){\small $B$}
					edge[e](Z)
					edge[e](L);
				\end{tikzpicture}
				\\ (j) & (k) & (l) \\
				\begin{tikzpicture}[scale=0.55]
					\node[q1](Z) at (5,0){\small $\Sigma$};                 \node[q1](L) at (2,0){\small $\Lambda$};
					\node[c1](S) at (0,-1.5){\small $X$};
					\node[c1](T) at (4,-1.5){\small $Y$}
					edge[e](S)
					edge[e](Z);
					\node[c1](X) at (0,1.5){\small $A$}
					edge[e](L)
					edge[e](S);
					\node[c1](Y) at (4,1.5){\small $B$}
					edge[e](L)
					edge[e](Z);
				\end{tikzpicture}
				&
				\begin{tikzpicture}[scale=0.55]
					\node[q1](Z) at (5,0){\small $\Sigma$};
					\node[q1](O) at (-1,0){\small $\Omega$};
					\node[q1](L) at (2,0){\small $\Lambda$};
					\node[c1](S) at (0,-1.5){\small $X$}
					edge[e](O);
					\node[c1](T) at (4,-1.5){\small $Y$}
					edge[e](S)
					edge[e](Z);
					\node[c1](X) at (0,1.5){\small $A$}
					edge[e](L)
					edge[e](O)
					edge[e](S);
					\node[c1](Y) at (4,1.5){\small $B$}
					edge[e](L)
					edge[e](Z);
				\end{tikzpicture}
				&
				\begin{tikzpicture}[scale=0.55]
					\node[q1](Z) at (5,0){\small $\Sigma$};
					\node[q1](O) at (-1,0){\small $\Omega$};
					\node[q1](L) at (2,0){\small $\Lambda$};
					\node[c1](S) at (0,-1.5){\small $X$}
					edge[e](O);
					\node[c1](T) at (4,-1.5){\small $Y$}
					edge[e](S)
					edge[e](Z);
					\node[c1](X) at (0,1.5){\small $A$}
					edge[e](L)
					edge[e](O);
					\node[c1](Y) at (4,1.5){\small $B$}
					edge[e](L)
					edge[e](Z);
				\end{tikzpicture}

				\\ (m) & (n) & (o) \\
				\begin{tikzpicture}[scale=0.55]
					\node[q1] (X) at (2,  2) {\small $\Lambda$};
					\node[q1] (Y) at (2,  0) {\small $\Omega$};
					\node[c1] (A) at (0, -2){\small $X$};
					\node[c1] (B) at (0, 0){\small $A$}
					edge[e](A)
					edge[e](Y);
					\node[c1] (C) at (0, 2){\small $Y$}
					edge[e](B)
					edge[e](X);
					\node[c1] (D) at (4, 1){\small $B$}
					edge[e](X)
					edge[e](Y);
				\end{tikzpicture}
				&
				\begin{tikzpicture}[scale=0.55]
					\node[q1] (Z) at (-2,  -1) {\small $\Sigma$};
					\node[q1] (X) at (2,  2) {\small $\Lambda$};
					\node[q1] (Y) at (2,  0) {\small $\Omega$};
					\node[c1] (A) at (0, -2){\small $X$}
					edge[e](Z);
					\node[c1] (B) at (0, 0){\small $A$}
					edge[e](Z)
					edge[e](A)
					edge[e](Y);
					\node[c1] (C) at (0, 2){\small $Y$}
					edge[e](B)
					edge[e](X);
					\node[c1] (D) at (4, 1){\small $B$}
					edge[e](X)
					edge[e](Y);
				\end{tikzpicture}
				&
				\begin{tikzpicture}[scale=0.55]
					\node[q1] (Z) at (-2,  -1) {\small $\Sigma$};
					\node[q1] (X) at (2,  2) {\small $\Lambda$};
					\node[q1] (Y) at (2,  0) {\small $\Omega$};
					\node[c1] (A) at (0, -2){\small $X$}
					edge[e](Z);
					\node[c1] (B) at (0, 0){\small $A$}
					edge[e](Z)
					edge[e](Y);
					\node[c1] (C) at (0, 2){\small $Y$}
					edge[e](B)
					edge[e](X);
					\node[c1] (D) at (4, 1){\small $B$}
					edge[e](X)
					edge[e](Y);
				\end{tikzpicture}
				\\ (p) & (q) & (r)
			\end{tabular}
			\caption{A maximal DAG has a set of \(d\)-separation relations unmatched by any latent-free DAG (and is thus \interesting by Theorem~\ref{thm:dsep_for_interesting}) if and only if it has one of these eighteen patterns as a subgraph.}
			\label{fig:d_sep_patterns_compact}
		\end{figure}
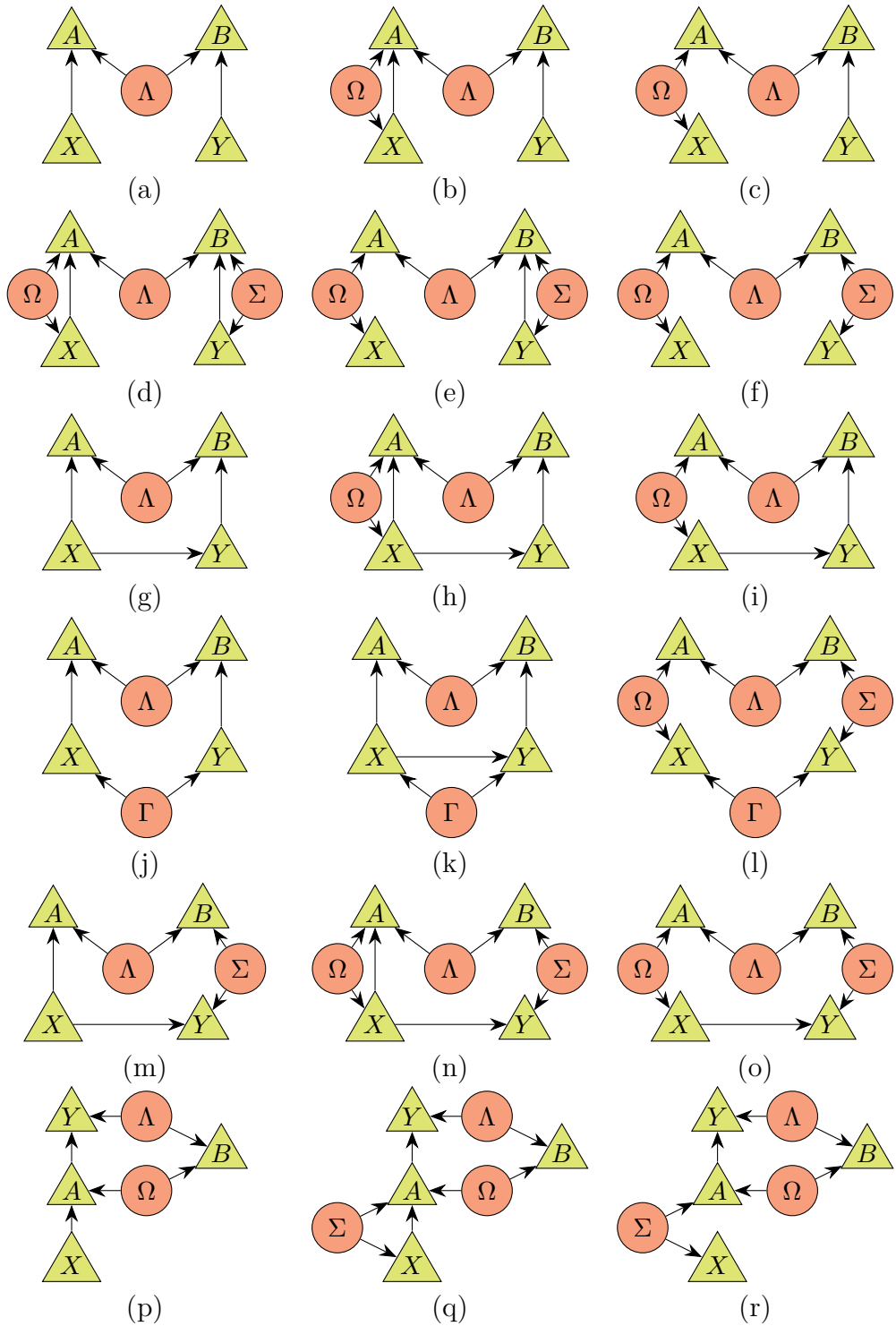
	\end{flushleft}

	\begin{theorem}[rapid \(d\)-separation condition for \interestingness]\label{thm:rapid_dsep_for_interesting}
		Let $G$ be an mDAG. If $G$ is not maximal, then it is \interesting by Theorem~\ref{thm:elie_esep}. If $G$ is maximal, then it has a set of observed \(d\)-separation relations unmatched by any latent-free DAG (\NLF) --- and is therefore \interesting pursuant to Theorem~\ref{thm:dsep_for_interesting} --- if and only if it contains one of the eighteen graph patterns (up to relabelling the nodes) presented in Figure~\ref{fig:d_sep_patterns_compact} as a subgraph.
	\end{theorem}
	\begin{proof}
		Evans'~\cite{Evans_2022} has shown that if a latent-permitting DAG $G$ has a \NLF set of observed \(d\)-separation relations, then either its associated PAG (partial ancestral graph) contains a so-called ``locally unshielded collider path" of length 3, and/or it contains a so-called ``discriminating path" of length 3.\footnote{The phrasing presented in the proof of Proposition 4.2 of Ref.~\cite{Evans_2022} says that either the associated PAG has a locally unshielded collider path of length exactly 3, all options of which are presented in Figure 4 of Ref.~\cite{Evans_2022}, or a discriminating path of length \emph{at least} 3. However, Proposition B.1 of~\cite{Evans_2022} further shows that all discriminating paths of length at least 3 either contain a locally unshielded collider path of length 3 and/or a discriminating path of length exactly 3, which is presented in Figure 5(i) of~\cite{Evans_2022}.} This implies that the PAG associated with $G$ will contain a sub-PAG matching one of the 4-node PAGs depicted within Figures 4 and 5(i) of~\cite{Evans_2022} whenever the observed \(d\)-separation relations of $G$ are \NLF.

		A PAG is an abstract graphical representation of a set of $d$-separation relations. In particular, two nodes are \emph{adjacent} in a PAG if and only the two nodes are \(d\)-unseparable in the original DAG. A ``path of length 3" in a PAG refers to a set of four nodes $\{X,A,B,Y\}$ such that $\{X,A\}$, $\{A,B\}$, $\{B,Y\}$ all constitute \(d\)-unseparable pairs.

		The adjacencies of a DAG $G$ and the adjacencies of the PAG associated with $G$ can, in general, differ. For \emph{maximal} DAGs, however, they must coincide: for those, adjacency is equivalent to \(d\)-unseparability. The PAG associated with a \emph{maximal} DAG $G$ will exhibit an unshielded collider path or a discriminating path if and only if in the original $G$ we can find four nodes $\{X,A,B,Y\}$ such that $\{X,A\}$, $\{A,B\}$, $\{B,Y\}$ represent adjacent pairs and such that the \(d\)-separation relations pertaining exclusively to $\{X,A,B,Y\}$ are of one of ``unshielded collider path" type or ``discriminating path" type. We do not define these two types for brevity, but we exhibit all such 4-node mDAGs in Figure~\ref{fig:d_sep_patterns_compact}.
	\end{proof}

	The set of \(d\)-separation relations of each one of the DAGs in Figure~\ref{fig:d_sep_patterns_compact} is:
	\begin{compactenum}
		\item $(A\perp Y|\emptyset)$  and $(A\perp Y|X)$ and $(B\perp X|\emptyset)$ and \\$(B\perp X|Y)$ and $(X\perp Y)$ and $(X\perp Y|A)$ and $(X\perp Y|B)$ \hfill[Figure~\ref{fig:d_sep_patterns_compact} (a)-(f)]
			\item $(A\perp Y|X)$ and $(B\perp X|Y)$ \hfill[Figure~\ref{fig:d_sep_patterns_compact} (g)-(k)]
			\item $(A\perp Y|\emptyset)$ and $(B\perp X|\emptyset)$ \hfill [Figure~\ref{fig:d_sep_patterns_compact} (l)]
			\item $(B\perp X|\emptyset)$ and $(A\perp Y|X)$\hfill [Figure~\ref{fig:d_sep_patterns_compact} (m)-(o)]
			\item $(B\perp X|\emptyset)$ and $(X\perp Y|A)$  \hfill [Figure~\ref{fig:d_sep_patterns_compact} (p)-(r)]
	\end{compactenum}

	By checking the \(d\)-separation of all the 4-observed-node mDAGs, we know that these are the only mDAGs that present these sets of \(d\)-separation relations. While searching for subgraphs is computationally easy, an alternative is to consider all 4-node subsets of a given large mDAG and ask if the \(d\)-separation relations \emph{pertaining exclusively to some four nodes} contains all and only one of the patterns listed above, up to relabelling. If yes, and if the large mDAG $G$ is maximal, then $G$ certainly contains one of the patterns of Figure~\ref{fig:d_sep_patterns_compact} as a subgraph.

	Figure~\ref{fig_nonmaximal} shows an example of a DAG that does not have any of the DAGs of Figure~\ref{fig:d_sep_patterns_compact} as a subgraph, but nevertheless has a \NLF set of \(d\)-separation relations. However, this is not a problem: this DAG is \emph{not} maximal ($X$ and $Y$ are not \(d\)-separable but are \(e\)-separable by the empty set), so its \interestingness follows from Theorem~\ref{thm:elie_esep}.

	\begin{figure}[h]
		\centering
		\begin{tikzpicture}[scale=0.8]
			\node[c](S) at (-1.5,2.5){$X$};
			\node[c](T) at (5.5,2.5){$Y$};
			\node[q](A) at (0,1){$\Lambda$};
			\node[q](B) at (4,1){$\Gamma$};
			\node[c](C) at (2,2){$C$}
			edge[e] (A)
			edge[e] (B);
			\node[c](E) at (4,4){$B$}
			edge[e](T)
			edge[e] (B)
			edge[e] (C);
			\node[c](D) at (0,4){$A$}
			edge[e](S)
			edge[e] (A)
			edge[e] (C);
		\end{tikzpicture}
		\caption{nonmaximal DAG. It does not have any of the DAGs of Figure~\ref{fig:d_sep_patterns_compact} as a subgraph, but it nevertheless has a \(d\)-separation pattern that does not correspond to any latent-free DAG.}
		\label{fig_nonmaximal}
	\end{figure}
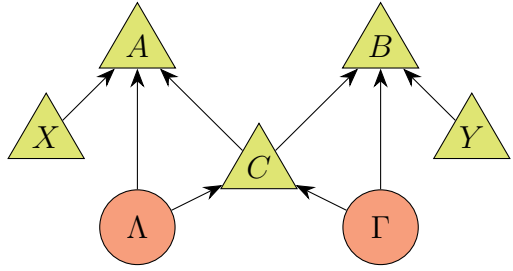

	Finally, we will show an explicit construction of a distribution which is in $\mathcal{I}_G$ but not in $\mathcal{C}_G$ for the maximal mDAGs which are shown \interesting by Theorem~\ref{thm:elie_esep}. First, we note that the Popescu-Rohrlich box support shown in Eq.~\eqref{eq_support_Bell} is not compatible with any of the mDAGs of Figure~\ref{fig:d_sep_patterns_compact}, as can be shown using Fraser's algorithm. This implies that the \emph{uniform} distribution over the events of the Popescu-Rohrlich box support is not classically compatible with any of the mDAGs of Figure~\ref{fig:d_sep_patterns_compact}, and is thus not an element of $\mathcal{C}_G$ for these mDAGs. On the other hand, it is well-known that said distribution obeys all of the conditional independence relations that come from the \(d\)-separation relations of the Bell DAG. As seen above, all the \(d\)-separation relations of the mDAGs of Figure~\ref{fig:d_sep_patterns_compact} are included in the \(d\)-separation relations of the Bell DAG, which implies that such distribution is an element of $\mathcal{I}_G$ for all the mDAGs of Figure~\ref{fig:d_sep_patterns_compact} .

	\section{Supports subsumes \(e\)-separation as a test of observational inequivalence}
	\label{appendix_proof_support_esep}
	Here we provide a full proof of Theorem~\ref{th_supp_subsume_inequivalence}, that says that if DAGs $G$ and $H$ can be shown inequivalent by virtue of having different sets of \(e\)-separation relations, then they can be shown inequivalent by virtue of having different sets of compatible supports at \emph{binary} variables. Before proving this more general result, we will prove the case for \(d\)-separation relations as an auxiliary lemma:

	\begin{lemma}[\(d\)-separation and compatible supports]
		\label{lemma:d-sep_supps}
		Let $G$ be an DAG with observed nodes $\boldsymbol{V}$, with $\boldsymbol{Z}\subseteq V$ being a subset of $\boldsymbol{V}$ and $X\in \boldsymbol{V}$ and $Y\in \boldsymbol{V}$ being two observed nodes. If the DAG $G$ \emph{does not} exhibit the \(d\)-separation relation $(X\perp Y|\boldsymbol{Z})$, then there is at least one support over binary variables which is compatible with $G$ but at the same time is in conflict with $(X\dbot Y|\boldsymbol{Z})$, i.e. is trivially incompatible with every other DAG for which $(X\perp Y|\boldsymbol{Z})$.
	\end{lemma}
	\begin{proof}
		We explicitly construct a support such that the marginal events $\{X{=}0,\boldsymbol{Z}=1\}$ and  $\{Y{=}1,\boldsymbol{Z}=1\}$ both occur, but such that the event $\{X{=}0,Y{=}1,\boldsymbol{Z}=1\}$ does not occur. Plainly such a support is \emph{trivially incompatible} with any DAG wherein $(X\perp Y|\boldsymbol{Z})$ per Lemma~\ref{lemma_trivial_support}. Here we show that a distribution with this support \emph{can} arise in every DAG wherein $(X\not\perp Y|\boldsymbol{Z})$. We note that this construction is identical to that which appears in Appendix E of \cite{finkelstein2021entropic}.

		If $(X\not\perp Y|\boldsymbol{Z})$ in $G$, then there exists a path from $X$ to $Y$ in $G$ which is unblocked by $\boldsymbol{Z}$. Every node of the path has either no parents in the path (in which case it is the base of a fork), one parent in the path (if it is the middle node of a chain or an end node of the path) or two parents in the path (if it is a collider). The support is constructed by assigning the following functional dependencies to the nodes of the path:
		\begin{itemize}
			\item A node $F\in\boldsymbol{F}$ (``F" for ``Fork") has 0 parents in the path:
			      \begin{equation*}
				      F= \begin{cases}
					      0 \text{ with probability } 1/2 \\
					      1 \text{ with probability } 1/2
				      \end{cases}
			      \end{equation*}
			\item A node $M\in\boldsymbol{M}$ (``M" for ``Mediary") has 1 parent $P_M$ in the path:
			      \begin{equation*}
				      M=P_M \text{ with unity probability}
			      \end{equation*}
			\item Node $C\in\boldsymbol{C}$ (``C" for ``Collider") has 2 parents $P_{C,1}$ and $P_{C,2}$ in the path:
			      \begin{equation*}
				      C= \begin{cases}
					      0 \text{ if } P_{C,1}\neq P_{C,2} \\
					      1 \text{ if } P_{C,1}=P_{C,2}
				      \end{cases}
			      \end{equation*}
		\end{itemize}

		Since the path is unblocked $\boldsymbol{C}\in\boldsymbol{Z}$ while $\boldsymbol{M}\boldsymbol{F}\not\in\boldsymbol{Z}$.

		This construction leads to a probability distribution compatible with $G$ wherein $P(\boldsymbol{M}=\boldsymbol{F}=0|\boldsymbol{Z}{=}1)=1/2$ and $P(\boldsymbol{M}=\boldsymbol{F}=1|\boldsymbol{Z}{=}1)=1/2$. Since $\{X,Y\}\subset \boldsymbol{M}\boldsymbol{F}$, we confirm that the marginal event $\{X{=}0,\boldsymbol{Z}=1\}$ occurs, and such that the marginal event $\{Y{=}2,\boldsymbol{Z}=1\}$ occurs, yet that the event $\{X{=}0,Y{=}1,\boldsymbol{Z}=1\}$ does not occur.
	\end{proof}

	Using Lemma~\ref{lemma:d-sep_supps}, we can now proceed to proof Theorem~\ref{th_supp_subsume_inequivalence}:

	\Suppesep*
	\begin{proof}
		Suppose that $G$ has an \(e\)-separation relation $(A\eperp B|\boldsymbol{C})_{del_{\boldsymbol{D}}}$ that $H$ does not have, where $A,B\in V$ are observed nodes and $\boldsymbol{C},\boldsymbol{D}\subseteq V$ are sets of observed nodes.

		By the definition of \(e\)-separation, we know that $G_{\boldsymbol{V}\setminus \boldsymbol{D}}$, the DAG obtained by deleting the nodes $\boldsymbol{D}$ from $G$, exhibits the \(d\)-separation relation $A\perp B|\boldsymbol{C}$ while $H_{\boldsymbol{V}\setminus \boldsymbol{D}}$ does not. This implies that \emph{none} of the supports compatible with $G_{\boldsymbol{V}\setminus \boldsymbol{D}}$ are in conflict with $A\dbot B|\boldsymbol{C}$, while from Lemma~\ref{lemma:d-sep_supps} we know that there is at least one support over binary variables which is compatible with $H_{\boldsymbol{V}\setminus \boldsymbol{D}}$ that is in conflict with $A\dbot B|\boldsymbol{C}$. This means that it is possible to find a support $\mathcal{S}_{\boldsymbol{V}\setminus \boldsymbol{D}}$ over binary variables which is compatible with $H_{\boldsymbol{V}\setminus \boldsymbol{D}}$ but not compatible with $G_{\boldsymbol{V}\setminus \boldsymbol{D}}$.

		Let $\mathcal{S}_{\boldsymbol{V}}$ be the support over $\boldsymbol{V}$ which coincides with $\mathcal{S}_{\boldsymbol{V}\setminus \boldsymbol{D}}$ for the variables in ${\boldsymbol{V}\setminus \boldsymbol{D}}$ and where the variables in $\boldsymbol{D}$ are set to a point value. This support is incompatible with $G$, because setting the variables in $\boldsymbol{D}$ to a point value has the same effect on its children as deleting the respective nodes (and thus reaching $G_{\boldsymbol{V}\setminus \boldsymbol{D}}$). Furthermore, $\mathcal{S}_{\boldsymbol{V}}$ is \emph{compatible} with $H$: since $\mathcal{S}_{\boldsymbol{V}\setminus\boldsymbol{D}}$ is compatible with $H_{\boldsymbol{V}\setminus \boldsymbol{D}}$, we can obtain $\mathcal{S}_{\boldsymbol{V}}$ by functional models where the children of $\boldsymbol{D}$ in $G$ ignore the parents in $\boldsymbol{D}$, as indicated in Lemma 13 of~\cite{finkelstein2021entropic}.

		Therefore, $\mathcal{S}_{\boldsymbol{V}}$ is a support over binary variables that is compatible with $H$ but incompatible with $G$.
	\end{proof}

\section{Results for mDAGs of 5 observed nodes}\label{app:5nodes}

    When trying to certify the \interestingness of mDAGs with 5 or more observed nodes, a simple but important consideration should be taken into account. Namely,
    \begin{proposition}\label{prop:nonalgebraic_subgraphs}
        Consider a DAG $G$, as well as the subgraph $G_{\text{del }\boldsymbol{S}}$ of $G$, where $G_{\text{del }\boldsymbol{S}}$ is formed by deleting a strict subset of $G'$s observed nodes from $G$. If $G_{\text{del }\boldsymbol{S}}$ is \interesting, then $G$ is \interesting as well.
    \end{proposition}
    \begin{proof}
        Proposition~\ref{prop:nonalgebraic_subgraphs} is an immediate consequence of Lemma~\ref{lem:point_distribution_subgraph} from the main text. That is, Lemma~\ref{lem:point_distribution_subgraph} ensure that $P(\boldsymbol{V}\setminus \boldsymbol{S})\not\in\mathcal{C}_{G_{\text{del }\boldsymbol{S}}}$ then $P(\boldsymbol{V})\not\in\mathcal{C}_G$ where ${{P(\boldsymbol{V})\coloneqq P(\boldsymbol{V}\setminus \boldsymbol{S})\delta_{\boldsymbol{S},0^{|\boldsymbol{S}|}}}}$. At the same time, if $P(\boldsymbol{V}\setminus \boldsymbol{S})\in\mathcal{I}_{G_{\text{del }\boldsymbol{S}}}$ then $P(\boldsymbol{V})\in\mathcal{I}_G$ where again ${{P(\boldsymbol{V})\coloneqq P(\boldsymbol{V}\setminus \boldsymbol{S})\delta_{\boldsymbol{S},0^{|\boldsymbol{S}|}}}}$. That $P(\boldsymbol{V})\in\mathcal{I}_G$ follows from the fact that no \emph{new} \(d\)-separation relations are induced on $\boldsymbol{V}\setminus \boldsymbol{S}$ by embedding the DAG $G_{\text{del }\boldsymbol{S}}$ as a subgraph of a larger DAG, namely $G$.
    \end{proof}

    The reader might ask why we did not employ Proposition~\ref{prop:nonalgebraic_subgraphs} in assessing the \interestingness of mDAGs with \emph{4} observed nodes. After all, there are five mDAGs with \emph{3} observed nodes. However, any 4-observed-nodes mDAGs which would be certifiable as \interesting via Proposition~\ref{prop:nonalgebraic_subgraphs} would already be picked up by Theorem~\ref{thm_setwise_nonmaximal}, since that theorem detects 100\% of the \interesting 3-observed-nodes mDAGs by itself.

    \begin{table}[h!]
        \centering
        \renewcommand{\arraystretch}{1.5}
        \begin{tabular}{|p{11cm}|c|}
            \hline
            \small Category                                                                             & \specialcell{\small mDAGs with \\\textbf{5} observed nodes} \\
            \hline
            {\small  Total Count}                                                                       & 1,718,596                      \\
            {\small remaining \# for which the \HLP criterion does not apply}                           & 1,009,961                      \\
            {\small remaining \# for which our nonmaximality condition does not apply}                  & 278,964                        \\
            {\small remaining \# for which our setwise nonmaximality condition does not apply}          & 118,278                        \\
            {\small remaining \# which do not contain an \interesting \mbox{4-observed-nodes} subgraph} & 12,834                         \\
            {\small remaining \# for which our \(d\)-separation condition does not apply}               & 12,834                         \\
            {\small remaining \# for which our \(e\)-separation condition does not apply}               & 12,834                         \\
            \hline
        \end{tabular}
        \caption{Results for mDAGs of 5 observed nodes.}\label{tab:append}
    \end{table}

	The application of the conditions for \interestingness presented here to mDAGs of 5 observed nodes gives the results shown in Table \ref{tab:append}. Although the application of Fraser's algorithm on the remaining 12,834 mDAGs was computationally infeasible, nevertheless, the application of all the other techniques provides a similar success percentage as compared to mDAGs of 4 and 3 observed nodes. Precisely, for mDAGs of 5 observed nodes all the other techniques (apart from supports) reduce the number of unresolved mDAGs of 5 observed nodes by 99.15\%, while this percentage is 99.89\% and 97.8\% for mDAGs of 4 and 3 observed nodes respectively. This result is consistent with the \HLP conjecture, though whether or not it can be considered \emph{evidence} in favor of the conjecture is debatable.

	Unsurprisingly, our  \(d\)-separation condition as articulated in Theorem~\ref{thm:dsep_for_interesting} is now effectively redundant to the conjuction of Proposition~\ref{prop:nonalgebraic_subgraphs} and Theorem~\ref{thm:elie_esep}, since a \emph{maximal} mDAG with 5+ observed nodes will have a set of observed \(d\)-separation relations inequivalent to any latent-free DAG only if the large mDAG contains one of 18 particular 4-observed-nodes mDAGs, as discussed extensively in Appendix~\ref{app_rapid_dsep}.

	The fact that the Theorem~\ref{new_e_sep} does not resolve any further mDAGs as \interesting is evidence in favor of a conjecture that Theorem~\ref{new_e_sep} is perhaps subsumed by the conjunction of Theorems~\ref{thm:elie_esep} and~\ref{thm:dsep_for_interesting}.

\end{appendices}
\end{document}